\documentclass[acmsmall,%
authorversion,
natbib=false
]{acmart}

\RequirePackage[
backend=biber,
style=ACM-Reference-Format,
language=american,
sortcites=true,
doi=false,
isbn=false,
date=year,
]{biblatex}

\DeclareSourcemap{
  \maps[datatype=bibtex]{
    \map{
      \step[fieldset=issn, null]
      \step[fieldset=pages, null]
      \step[fieldset=location, null]
    }
  }
}

\usepackage{subfig}
\usepackage{graphics} %
\usepackage{epsfig} %
\usepackage{booktabs}

\usepackage{algorithm}%
\usepackage{algpseudocode}%
\usepackage{siunitx}
\usepackage{import} %
\usepackage{color}
\usepackage{xcolor}
\usepackage{transparent}
\usepackage{svg}
\usepackage{multirow}
\usepackage{tikz}

\usepackage[english]{babel}
\usepackage{blindtext}
\usepackage{csquotes}

\usepackage{tabularx}
\usepackage{booktabs}

\usepackage{pgfplots}
\usepgfplotslibrary{groupplots,dateplot}
\usetikzlibrary{patterns,shapes.arrows}
\pgfplotsset{compat=newest}

\usepgfplotslibrary{external} 
\def\clap#1{\hbox to 0pt{\hss#1\hss}}

\usepackage{xspace}
\newcommand*{\eg}{e.g.\@\xspace}
\newcommand*{\ie}{i.e.\@\xspace}

\usepackage{tikz}
\usetikzlibrary{positioning}
\usetikzlibrary{scopes}
\usetikzlibrary{calc}
\usetikzlibrary{decorations.pathmorphing}
\usetikzlibrary{decorations.pathreplacing}
\usetikzlibrary[shapes.symbols]
\usetikzlibrary[backgrounds]
\usetikzlibrary[fit]

\definecolor{mplblue}{RGB}{31,119,180}
\newcommand*\circled[1]{\tikz[baseline=(char.base)]{\node[shape=circle,draw,inner sep=1pt, color=mplblue, fill=white] (char) {\footnotesize #1};}}

\newtheorem{lemma}{Lemma}

\newtheorem{theorem}{Theorem}
\newtheorem{definition}{Definition}

\setcopyright{acmcopyright}
\copyrightyear{2022}
\acmYear{2022}
\acmDOI{10.1145/3520440}

\acmJournal{TOIT}
\acmVolume{22}
\acmNumber{4}
\acmArticle{97}
\acmMonth{8}

\begin{document}

\title{Optimization-Based Predictive Congestion Control for the Tor Network: Opportunities and Challenges}

\author{Christoph Döpmann}
\authornote{Both authors contributed equally to this research.}
\email{christoph.doepmann@tu-berlin.de}
\affiliation{%
  \institution{Technische Universität Berlin}
  \department{Distributed Security Infrastructures}
  \city{Berlin}
  \country{Germany}
}
\author{Felix Fiedler}
\authornotemark[1]
\email{felix.fiedler@tu-dortmund.de}
\affiliation{%
  \institution{Technische Universität Dortmund}
  \department{Process Automation Systems}
  \city{Dortmund}
  \country{Germany}
}
\author{Sergio Lucia}
\email{sergio.lucia@tu-dortmund.de}
\affiliation{%
  \institution{Technische Universität Dortmund}
  \department{Process Automation Systems}
  \city{Dortmund}
  \country{Germany}
}
\author{Florian Tschorsch}
\email{florian.tschorsch@tu-berlin.de}
\affiliation{%
  \institution{Technische Universität Berlin}
  \department{Distributed Security Infrastructures}
  \city{Berlin}
  \country{Germany}
}

\begin{abstract}
Based on the principle of onion routing,
the Tor network achieves anonymity for its users
by relaying user data over a series of intermediate relays.
This approach makes congestion control in the network a challenging task.
As of today, this results in higher latencies due to considerable backlog
as well as unfair data rate allocation.
In this paper, we present a concept study of PredicTor,
a novel approach to congestion control that tackles clogged overlay networks.
Unlike traditional approaches, it is built upon the idea of distributed model predictive control,
a recent advancement from the area of control theory.
PredicTor is tailored to minimizing latency in the network and achieving max-min fairness.
We contribute a thorough evaluation of its behavior in both
toy scenarios to assess the optimizer
and complex networks to assess its potential.
For this, we conduct large-scale simulation studies
and compare PredicTor to existing congestion control mechanisms in Tor.
We show that PredicTor is highly effective in reducing latency
and realizing fair rate allocations.
In addition, we strive to bring the ideas of modern control theory to the networking community,
enabling the development of improved, future congestion control.
We therefore demonstrate benefits and issues alike
with this novel research direction.
\end{abstract}




\maketitle

\section{Introduction}

In today's digital society, protecting Internet privacy has become more important than ever.
The growing demand for online anonymity has resulted in advanced technical solutions to satisfy this need.
With around 2.2~million daily users~\cite{tor-metrics},
the Tor network~\cite{dingledine2004tor} is by far the most widely used anonymization network, as of today.
Tor builds upon the idea of \emph{onion routing}~\cite{goldschlag1996onionrouting}.
It consists of an overlay network connecting so-called relay nodes,
which can be used to establish anonymous connections.
To this end, the Tor client software builds a cryptographically-secured \emph{circuit},
a path over three relays, where each relay knows its immediate neighbors only.

This has several performance implications.
Due to re-routing the traffic multiple times through the overlay network,
an extra delay is inevitable to gain anonymity.
The performance---in terms of latency, data rates, and fairness---is however
suboptimal~\cite{DBLP:conf/uss/ReardonG09,DBLP:conf/p2p/DhungelSRHR10}.
One of the major shortcomings is the lack of
effective congestion control~\cite{DBLP:conf/pet/AlSabahBGGMSV11,DBLP:conf/p2p/DhungelSRHR10}
that minimizes network load and optimizes the user-perceivable performance.
While congestion control is a nontrivial task even for single connections,
relaying data over a series of nodes, like in Tor, amplifies the problem;
especially when rising delays occur in the network.
In particular, Tor relays are unable to react to congestion,
for example by signaling upstream to throttle sending rates.
Moreover, a growing demand for the Tor network will also result in
a growing need for effective congestion control to satisfy the users' expectations
as far as latency, throughput, and fairness are concerned.
Therefore, one also has to consider novel research approaches to meet these challenges.

With \emph{PredicTor}~\cite{PredicTor},
we introduce a new research direction towards congestion control
in multi-hop overlay networks like the Tor network.
PredicTor is the first system to apply
\emph{distributed Model Predictive Control}~(MPC)~\cite{Mota2012} to congestion control
in the Tor network.
MPC in general is a modern technique from the field of control theory
that uses predictions about the future system state as well as the repeated solving
of a formal optimization problem to achieve optimal behavior.
Predictions are deduced from a mathematical system model
that is instantiated with real-world measurements.
For applying MPC in the context of multi-hop congestion control,
we \emph{distribute} it among relays.
That is, each relay solves the optimization problem with its \emph{local} view of the network,
but controllers%
\footnote{The term \enquote{controller} refers to the local application of control techniques.
It does not imply a centralized entity.}
cooperate by exchanging their predictions to establish network-wide behavior.
In contrast to the current behavior of Tor,
PredicTor avoids congestion by generating \emph{backpressure}.
By relying on a formal definition of the optimization goal,
it becomes possible to optimize the congestion control
within the network for specific optimization objectives.
In PredicTor, we put a special emphasis on low latency and fairness in the network,
because these have previously been identified
to be especially problematic in the Tor network~\cite{tschorsch11maxmin,DBLP:conf/uss/ReardonG09}.
While optimization-based rate allocation has been researched before,
with equivalent formulations for TCP and other methods~\cite{He2007},
we introduce a novel optimization-based \emph{max-min fairness} formulation.

In addition to presenting PredicTor,
the goal of this paper is to bring the underlying
control-theoretic approach to the network community.
We pinpoint the merits and the potential of applying distributed MPC to congestion control,
but also point out current shortcomings thereof.
Our work should be understood as a concept study for this novel field
rather than a ready-to-deploy finished technical solution.
We envision opening up new directions and fostering the development
of novel, innovative techniques for congestion control.
Our evaluation reveals that PredicTor is able to clearly reduce latency:
In a small model scenario, it achieves a latency reduction from 553~ms (vanilla Tor) to 94~ms.
In larger, random networks, the advantage becomes even more apparent because,
in contrast to traditional approaches,
latency does not significantly grow with growing congestion.
At the same time, PredicTor consistently realizes near-perfect max-min fairness.
However, we show that this comes at the cost of lower throughput and more signaling overhead. %

The contributions are summarized as follows.

\begin{itemize}
\item We introduce \emph{PredicTor} for congestion control in the Tor network based on distributed MPC.
      Comparing to~\cite{PredicTor}, we add an important change
      to its optimization problem that strengthens its robustness.
\item We present a novel, optimization-based formulation of max-min fairness
      and leverage it as an optimization goal in PredicTor.
\item We implement a prototype of PredicTor to enable experimental assessment of its behavior.
      Our implementation is made available as an open source software project.\footnote{\texttt{https://github.com/cdoepmann/PredicTor}}
\end{itemize}

PredicTor was first introduced in~~\cite{PredicTor}.
In this journal paper, we especially focus on PredicTor's significance for networking research by covering the following additional aspects.

\begin{itemize}
\item We introduce optimization-based predictive congestion control to the networking community.
\item We discuss possible security and privacy implications of PredicTor.
\item We provide a simulation study of PredicTor's performance in complex network scenarios,
      analyzing whether its optimization goals can be realized in non-trivial environments.
\item We leverage the results as well as a thorough investigation of PredicTor's
      underlying assumptions to identify the benefits as well as potential drawbacks
      of this new approach towards congestion control.
\end{itemize}

This paper is structured as follows:
We start by presenting the preliminaries for our approach,
including related terminology and mathematical notation, in Section~\ref{sec:preliminaries}.
Section~\ref{sec:mpc_formulation} introduces PredicTor itself in full detail,
including our novel optimization-based method for obtaining max-min fairness
and the dynamic system model.
We discuss security and privacy implications of PredicTor in Section~\ref{sec:security}.
In Section~\ref{sec:eval}, we present our evaluation of PredicTor,
focusing first on small scenarios to demonstrate its functioning
before we leverage larger-scale simulations of complex networks for deeper insights.
In this section, we also discuss the implications for further research on congestion control using MPC.
We complete our work by presenting related work in Section~\ref{sec:related-work}
and summarizing this contribution in Section~\ref{sec:conclusion}.

\section{Preliminaries} \label{sec:preliminaries}

\subsection{The Tor Network}

Facing today's growing need for online privacy,
Tor denotes an essential tool for applications that require anonymity on the Internet.
It is an \emph{overlay network} that makes use of
the principle of \emph{onion routing}~\cite{goldschlag1996onionrouting}.
It achieves anonymity by tunneling users' data through the network,
over a series of relays, called a \emph{circuit}.
Onion routing ensures that each hop in the relay only knows its immediate predecessor and successor.
As a consequence, destinations cannot identify the origin of streams of communication they receive.
Clients typically choose a sequence of three random relays for constructing a circuit.
The necessary resources (service and bandwidth) are contributed by volunteers
and are not subject to a central authority.

More formally, we introduce Tor as an overlay network graph~$G(N,E)$
where $N$ denotes the set of nodes and $E$ the set of overlay links.
The network has a total of $|N|=n$~nodes and $|E|=e$~connections.
We denote the set of Tor circuits~$P$ with $i \in P$
being the i-th circuit of the set of cardinality $|P|=p$.
$P_{\alpha} \in P$ denotes the subset of circuits traversing node~$\alpha \in N$.
Generally, we refer to circuits with Roman letters and to nodes with Greek letters.
When considering the network at the circuit level,
we denote the data rate of circuit~$i$ with $r_i$ (in packets per second).
Furthermore, each node~$\alpha \in N$ of the overlay network
has a limited capacity~$C_{\alpha}$, since overlay connections share the same physical connection.
Each node~$\alpha \in N$ can receive, store, and send data from each circuit~$i \in P_{\alpha}$.
We denote $s_{\alpha, i} $ the circuit queue (storage in number of packets)
in node~$\alpha$ for circuit~$i$
and the vector with all queues for each circuit
in node~$\alpha$ as $s_{\alpha} \in \mathbb{N}^{|P_{\alpha}|}$.

A useful metric to measure the congestion in the network is given by the \emph{backlog} that captures the amount of data that is on its way through the network. In terms of our formal description, we define it as follows:
\begin{definition}
    The data backlog~$b$ of a network~$G(N,E)$ is computed for all nodes~$\alpha \in N$ and all circuits~$i\in P$ as:
    \begin{equation}
        b = \sum_{\alpha \in N} \sum_{i \in P} s_{\alpha,i}.
    \end{equation}
\end{definition}

  \definecolor{color0}{rgb}{0.12156862745098,0.466666666666667,0.705882352941177}
  \definecolor{color1}{rgb}{1,0.498039215686275,0.0549019607843137}
  \definecolor{color2}{rgb}{0.172549019607843,0.627450980392157,0.172549019607843}
    \tikzset{
    written/.style={text=black,font=\scriptsize},
    circuit/.style={ultra thick},
    }
    \def\onion{\includegraphics[height=3.0em]{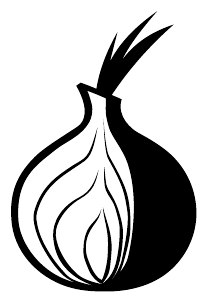}}
    \def\xfactor{1.0}
    \def\drawtopology{%
        \node (bottleneck) {\onion};
        \node[left=\xfactor*7em of bottleneck] (sender 1) {\onion};
        \node[below=0.5em of sender 1] (sender 2) {\onion};
        \node[right=\xfactor*7em of bottleneck] (receiver 2) {\onion};
        \node[above=0.5em of receiver 2] (receiver 1) {\onion};
        \node[below=0.5em of receiver 2] (receiver 3) {\onion};

        \begin{scope}[on background layer]
        \draw[->,circuit,color=color0,transform canvas={yshift=-6pt}]
            ($(sender 1.center) + (-4em*\xfactor,+6pt)$) --
            ($(bottleneck.center) + (0,+6pt)$) --
            node[above,sloped,written] {circuit 1}
            (receiver 1.center) --
            ++(4em*\xfactor,0);
        \draw[->,circuit,color=color1,transform canvas={yshift=-6pt}]
            ($(sender 1.center) + (-4em*\xfactor,0)$) --
            (sender 1.center) --
            (bottleneck.center) --
            node[above,sloped,written] {circuit 2}
            (receiver 2.center) --
            ++(4em*\xfactor,0);
        \draw[->,circuit,color=color2,transform canvas={yshift=-6pt}]
            ($(sender 2.center) + (-4em*\xfactor,0)$) --
            (sender 2.center) --
            ($(bottleneck.center) + (0,-6pt)$) --
            node[above,sloped,written] {circuit 3}
            (receiver 3.center) --
            ++(4em*\xfactor,0);
        \end{scope}
    }
\begin{figure}
    \centering
    \begin{tikzpicture}
        \drawtopology

        \node[above=-5pt of bottleneck,written] {bottleneck};

        \node[fit={($(sender 2.south west) + (-0.5em,0.5em)$) ($(receiver 1.north east) + (0.5em,0)$)}] (container) {};
        \draw[dashed,thick] (container.north west) -- (container.south west);
        \draw[dashed,thick] (container.north east) -- (container.south east);
        \node[below left=3pt of container.north west,written] {\bfseries senders};
        \node[below right=3pt of container.north east,written] {\bfseries receivers};
    \end{tikzpicture}
    \caption{Example Tor topology (toy scenario).}
    \label{fig:example-topology}
\end{figure}
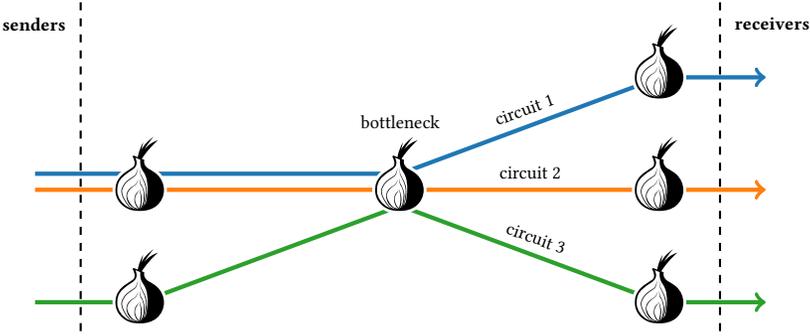

For various parts of this paper, we focus on a small example topology of circuits,
depicted in Figure~\ref{fig:example-topology}.
It consists of two sending relays that carry three circuits,
whose traffic streams meet in a shared node, constituting a bottleneck.
Afterwards, the three circuits go to distinct destination relays.
We found this simple topology to be useful
for evaluating the behavior of congestion control in Tor
because it represents a typical situation in which a single relay is overloaded by several circuits.
The scenario is still simple enough to understand the decisions made by congestion control.

\subsection{Fairness}

When considering performance in networks,
the immediate measures that come into mind are throughput and latency.
However, another important metric is fairness.
Especially in overlay networks like Tor, where many data transfers compete for the available resources,
fairness is important to guarantee an adequate user experience to the majority of users.
Different fairness measures have been put forward~\cite{bertsekas1992data}.
We here focus on the strong notion of \emph{max-min fairness}
as it has been proposed as fairness goal for the Tor network.
To define it formally, we first introduce the notion of \emph{feasibility}---%
that is, distributions of data rates that can actually be realized in the network~\cite{bertsekas1992data}.
\begin{definition}
	\label{def:feasible_rate}
	A rate vector $r=[r_1, r_2, \dots, r_p]$ is feasible if:
	\begin{align}
	\forall i \in P:& \quad    0\leq r_i \quad  \text{and}\\
	\forall \alpha \in N:& \quad    \sum_{i\in P_\alpha} r_i \leq C_{\alpha}.
	\label{eq:feasible_rate_capacity_limit}
	\end{align}
	We denote $R_f$ the set of feasible rate vectors.
\end{definition}
This allows us to define max-min fairness as follows:
\begin{definition}
\label{def:max_min_fair}
A feasible rate vector~$r^f\in R_f$ is called max-min fair,
if for all circuits~$i\in P$ and for all other feasible rates~$\bar{r}\in R_f$ it holds that:
\begin{equation}
	\begin{gathered}
	     \bar{r}_i \geq r_i^f \Rightarrow \exists \, j\in P:  r_{j}^f \leq r_i^f \land \bar{r}_j \leq r_j^f.
	\end{gathered}
\end{equation}
\end{definition}
This definition means that if a rate~$r^f$ is max-min fair,
any other feasible rate that increases the rate for the favored circuit~$i$
comes at the cost of reducing the rate for the disadvantaged circuit~$j$,
which is already smaller than the rate of circuit~$i$.

\subsection{Model Predictive Control}\label{ssec:MPC}
In this subsection, we briefly review the concept of model predictive control~(MPC),
which is used to formulate our proposed congestion controller.
Fundamental for this approach is the notion of a \emph{dynamic system}:
\begin{equation}\label{eq:dynamic_system}
	x^{k+1} = f(x^k, u^k, p^k),
\end{equation}
which relates a state $x^k \in \mathbb{R}^n$,
input $u^k \in \mathbb{R}^m$,
and parameter~$p^k\in \mathbb{R}^p$ at sampling time~$k$
to the state at the next sampling time $k+1$.
Under the assumption that the model~$f(x^k,u^k, p^k)$ accurately describes the system,
Equation~\eqref{eq:dynamic_system} can be used to compute the future states of the system given the initial state,
sequence of inputs and parameters.
Based on the dynamic system in \eqref{eq:dynamic_system}, we introduce the finite horizon optimal control problem (OCP):
\begin{subequations}\label{eq:MPC_problem_general}
\begin{align}
\label{eq:MPC_problem_general_01}
\min_{\textbf{u}, \textbf{x}}\quad &\sum_{k=0}^{N_{\text{horz}}}
l(x^k,u^k,p^k)\\
\text{\textbf{subject to}:}\quad && \nonumber\\
\text{state dynamics:} \quad & x^{k+1} = f(x^k, u^k, p^k),
\label{eq:MPC_problem_general_02}\\
\text{constraints:}\quad  &  g(x^k,u^k,p^k) \leq 0 \quad       \forall k =0,\dots, N_{\text{horz}}
\label{eq:MPC_problem_general_03}\\
\text{initial conditions:}\quad  & x^0 = x_{\text{init}}.
\end{align}
\end{subequations}
In this problem, we are optimizing over finite sequences of inputs $\mathbf{u}=[u^0,\dots, u^{N_{\text{horz}}}]$
and states $\mathbf{x}=[x^0,\dots, x^{N_{\text{horz}}+1}]$.
The optimal solution is obtained for a given initial state $x_{\text{init}}$
and a sequence of parameters $\mathbf{p}=[p^0,\dots, p^{N_{\text{horz}}}]$.
Note that we use bold letters to denote trajectories.
The objective is to minimize an arbitrary cost function under the consideration of additional constraints.
Typically, this cost function consists of individual contributions for each step of the horizon,
as shown in~\eqref{eq:MPC_problem_general_01}.
The previously introduced dynamic system \eqref{eq:dynamic_system} is considered in~\eqref{eq:MPC_problem_general_02}
as an equality constraint.
Additionally, we have in~\eqref{eq:MPC_problem_general_03} inequality constraints on states and inputs, possibly under consideration of the parameters.
This possibility to explicitly formulate constraints
is a major advantage of MPC over alternative advanced control techniques.

As the solution of the OCP (see~\eqref{eq:MPC_problem_general}),
we obtain the predicted future sequence of states and the respective sequence of inputs.
For the control application, the first element of the sequence of inputs, \ie $u^0$, is applied to the system,
typically in the form of a constant value over a finite sampling time.
After this sampling time, the new state of the system is obtained and together with the updated sequence of parameters
problem the OCP in~\eqref{eq:MPC_problem_general} is solved again.
Feedback trough this \emph{closed-loop} application allows to robustly react to disturbances
and mitigates potential mismatches between model and controlled system.

MPC is also a popular method to deal with distributed control systems.
In this application, multiple controllers make local decisions and attempt to achieve global control goals
by communicating their decisions.
In particular, the distributed MPC controllers can exchange their predicted future states and inputs
which are obtained as a byproduct when computing the current input to the system.
Connected controllers can consider this information as the additional parameters in~\eqref{eq:MPC_problem_general}.
Knowledge over future actions of connected controllers has the significant advantage that the effect of delay can be mitigated.

\section{PredicTor}\label{sec:mpc_formulation}
In this section, we introduce PredicTor, our newly proposed congestion controller for the Tor network.
PredicTor is developed with the following objectives in mind:
Primarily, we are aiming to avoid congestion by limiting the data backlog of circuits,
and secondly, we want to achieve \emph{max-min fairness} of the network.
To this end, we first present an optimization-based method to obtain max-min fairness
of an overlay network (Theorem~\ref{theo:max_min_optim}) in Subsection~\ref{ssec:optim_fairness},
which we have previously derived in~\cite{PredicTor}.
However, the presented Theorem~\ref{theo:max_min_optim} cannot directly be used for congestion control,
as it would require global knowledge and control authority of the Tor network.
Instead, it serves as the basis for our proposed distributed MPC formulation
for which we establish the preliminaries in Subsection~\ref{ssec:feedback_distr_mpc}.
In particular, we introduce the states, inputs, and system dynamics as well as the
concept of information exchange between adjacent nodes.
In comparison to our previous work~\cite{PredicTor},
this concept has been extended to address several shortcomings in previously unconsidered situations.
Most importantly, PredicTor nodes are now capable to request an exact rate increase from their successor nodes.
The full optimal control problem is then stated in Subsection~\ref{ssec:ocp}.
Finally, we discuss the interaction of PredicTor and a Tor relay in Subsection~\ref{ssec:controller_tor_interact}.

\subsection{Optimization-based Fairness}\label{ssec:optim_fairness}
We present an optimization-based method (Theorem~\ref{theo:max_min_optim}) to achieve max-min fairness.
For this, we first introduce the formal notion of a \emph{bottleneck}.
\begin{definition}
	\label{def:bottleneck}
	For a circuit~$i\in P_{\alpha}$ and a rate vector~$r$, we denote node~$\alpha \in N$ a bottleneck, if:
	\begin{equation}
	\sum_{i \in P_{\alpha}} r_i = C_{\alpha}, \quad \forall j \in P_{\alpha}: \ r_i \geq r_j
	\end{equation}
\end{definition}
\begin{lemma}
	\label{lemma:min_max_bottleneck}
	Let $r^f$ be a max-min fair rate vector. Each circuit~$i \in P$ has exactly one bottleneck.
	This bottleneck is the global rate-limiting factor of the circuit under stationary conditions.
\end{lemma}
\begin{proof}
	The proof is shown in~\cite{bertsekas1992data}.
\end{proof}
This allows us to state the following theorem, which we previously introduced in~\cite{PredicTor}.
\begin{theorem}
\label{theo:max_min_optim}
An overlay network achieves max-min fairness with rate~$r = r^{\text{max}} - \Delta r$ as the optimal solution of:
\begin{equation}
    \label{eq:opt_max_min_fairness}
    \begin{aligned}
      c= \min_{\Delta r} \sum_{i \in P}& \Delta r_i^2\\
        \text{subject to:}\quad
         r^{\text{max}} - \Delta r&\in R_f,\\
        0\leq \Delta r &\leq r^{\text{max}},
    \end{aligned}
\end{equation}
where $\Delta r$ is an auxiliary variable that can be interpreted as the unused rate with respect to the arbitrary upper limit~$r^{\text{max}}$ which must satisfy
$r^{\text{max}}\geq\max(C_1, C_2, \dots, C_n)$.
\end{theorem}
\begin{proof}
	The proof is presented in~\cite{PredicTor}.
\end{proof}
Theorem~\ref{theo:max_min_optim} thus allows us
to obtain the global \emph{max-min fair} rate $r$ of an overlay network
as the solution of a convex optimization problem.
Intuitively, the formulation in Theorem~\ref{theo:max_min_optim} works because we are minimizing the rate that is not allocated, with respect to some arbitrary upper limit and under consideration of feasible rates. The quadratic term results in fairness because it is always desirable to allocate a higher rate (\ie, reduce $\Delta r$) to the circuit with the smallest rate (\ie, with the highest $\Delta r$).

\subsection{Distributed MPC}\label{ssec:feedback_distr_mpc}
In this subsection, we present the preliminaries for the statement of the PredicTor optimal control problem.
In particular, we define states, inputs, and the dynamic system equation
and introduce our concept for distributed MPC.
This includes the question which information is exchanged and how it is incorporated
into the optimal control problem to achieve our previously defined control goals.
For the interaction of multiple nodes, we denote $\alpha \in N$ the currently considered node, with connections to predecessor ($\beta$) and successor ($\gamma$) nodes.
For the current node~$\alpha$, it is irrelevant whether the incoming data comes from several nodes or only from a single node.
To simplify the notation, we assume that all incoming data (even for different circuits)
comes from a single predecessor node~$\beta$ and is forwarded to a single successor node~$\gamma$.
The interaction of multiple controllers is illustrated in Figure~\ref{fig:feedback_exchange}
and will be discussed in the following.

\begin{figure}
	\footnotesize
	\centering
	\includegraphics[width=1\linewidth]{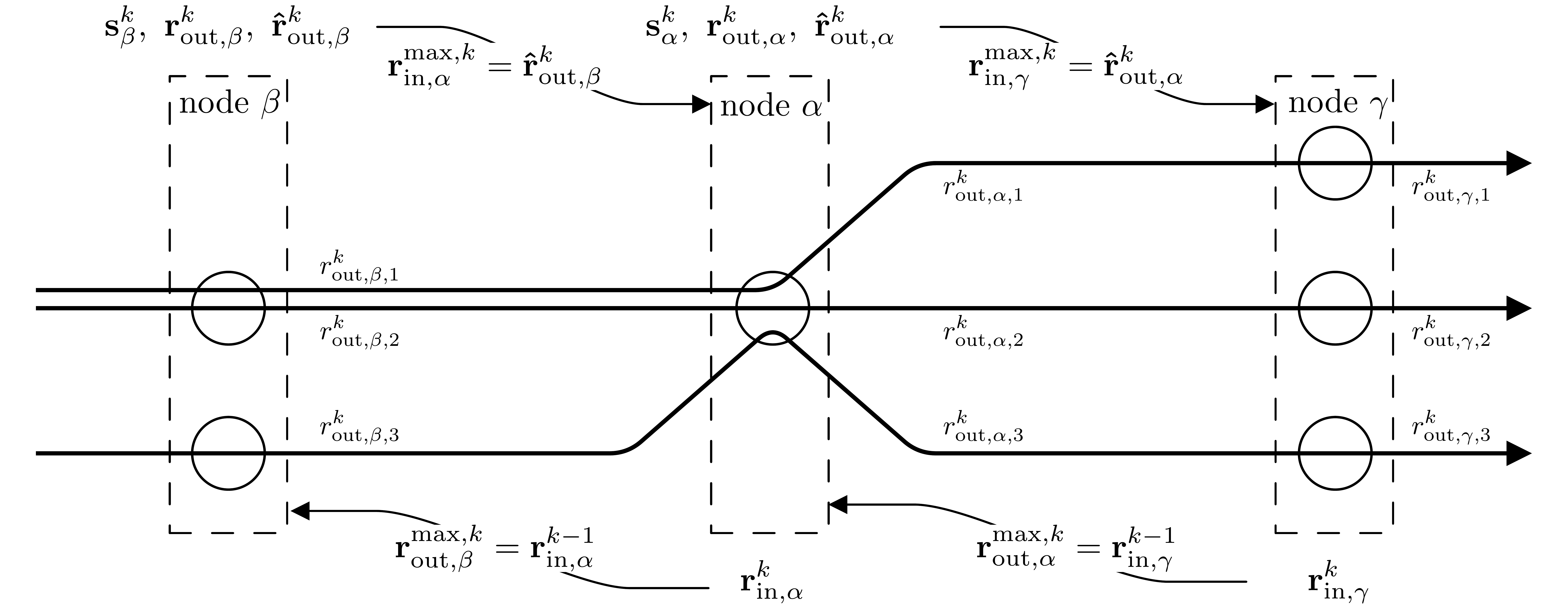}
	\caption{Information exchange (predicted future trajectories in bold font) of node~$\alpha$ with adjacent nodes.
		From the perspective of $\alpha$ all predecessor nodes are summarized as $\beta$ and all successor nodes as $\gamma$,
		for the sake of a concise notation.
		The rates at which data is sent at time~$k$ for circuit~$i$ at node~$\alpha$ is $r_{\text{out},\alpha, i}^k$.
	}
	\label{fig:feedback_exchange}
\end{figure}

The controller at node $\alpha$ takes local decisions regarding incoming and outgoing rates,
under consideration of predicted future actions from the adjacent nodes.
Predictions are obtained on the basis of dynamic models in the form of \eqref{eq:dynamic_system}.
We first introduce the state~$s_{\alpha}^k$ that denotes the queue size
for all circuits $i\in P_{\alpha}$ in node~$\alpha$ and at time step $k$.
The dynamic model equation can be written as:
\begin{align}
\label{eq:mpc_model}
s_{\alpha}^{k+1} &= s_{\alpha}^k + \Delta t(r_{\text{in}, \alpha}^k - r_{\text{out}, \alpha}^k),
\end{align}
where $\Delta t$ denotes the sampling time.
As inputs in \eqref{eq:mpc_model}, we introduce the incoming $r_{\text{in},\alpha}$ and outgoing $r_{\text{out},\alpha}$ rate.
For the optimal control problem, we first introduce trivial constraints for the queue size:
\begin{equation}\label{eq:cons_queue_size}
	0\leq s_{\alpha}^k \leq s_{\alpha}^{\text{max}},
\end{equation}
for the rates:
\begin{subequations}
\begin{align}
	\label{eq:predictor_r_in_positive}
	0 &\leq r_{\text{in},\alpha}^k\\
	\label{eq:predictor_r_out_positive}
	0 &\leq r_{\text{out},\alpha}^k,
\end{align}
\end{subequations}
and for the link capacities:
\begin{subequations}\label{eq:predictor_link_capacity_constraints}
\begin{align}
	\label{eq:predictor_link_capacity_constraints_in}
	\sum_{i\in P_{\alpha}}  &r_{\text{in},\alpha,i}^k  \leq C_{\alpha}^{\text{in}}\\
	\label{eq:predictor_link_capacity_constraints_out}
	\sum_{i\in P_{\alpha}}  & r_{\text{out},\alpha,i}^k \leq C_{\alpha}^{\text{out}}.
\end{align}
\end{subequations}
Furthermore, we have additional constraints regarding the incoming and outgoing rates which depend
on the actions of the adjacent nodes.
From the successor node $\gamma$, we receive information about $\mathbf{r}_{\text{in},\gamma}^k$
which limits the local outgoing rate as follows:
\begin{equation}
	\label{eq:mpc_pred_r_out_max}
	{r}_{\text{out},\alpha}^{k} \leq
	{r}_{\text{out},\alpha}^{\text{max},k}
	={r}_{\text{in},\gamma}^{k-1}.
\end{equation}
Note that we consider $\mathbf{r}_{\text{in},\gamma}^{k-1}$
(the successor's information from the \emph{previous} time step $k-1$)
since the information is delayed.
The constraints for the incoming rate $\mathbf{r}_{\text{in}, \alpha}^k$ are considerably more complex.
From the source node,
we receive the predicted outgoing rate $\mathbf{r}_{\text{out},\beta}^k$,
the predicted circuit queue $\mathbf{s}_{\beta}^k$
and a virtual outgoing rate $\mathbf{\hat r}_{\text{out},\beta}^k$.
With this information, we constrain $\mathbf{r}_{\text{in}, \alpha}^k$ in two ways.
First, we introduce an additional variable $\tilde{s}_{\alpha|\beta}^{k}$, which denotes
the estimated queue size of node $\beta$ from the perspective of node $\alpha$ at time $k$.
To express $\tilde{s}_{\alpha|\beta}^{k}$, we introduce the state variable $\Delta s_{\alpha|\beta}$, such that
\begin{subequations}
	\label{eq:mpc_pred_predecessor}
	\begin{align}
		\tilde{s}_{\alpha|\beta}^{k} &= s_{\beta}^k - \Delta s_{\alpha|\beta}^{k},
		\intertext{with the dynamic system equation in the form of \eqref{eq:dynamic_system}:}
		\Delta s_{\alpha|\beta}^{k+1} &= \Delta s_{\alpha|\beta}^{k} +\Delta t( r_{\text{in},\alpha}^k - r_{\text{out},\beta}^k).
		\label{eq:mpc_update_eq_delta_s}
	\end{align}
\end{subequations}%
Note that in contrast to \eqref{eq:mpc_pred_r_out_max}, we consider information from the source node $\beta$ at the current time-step $k$,
as both the control action and the information are delayed.
Equation~\eqref{eq:mpc_pred_predecessor} states that any value $r_{\text{in},\alpha}^k\neq r_{\text{out},\beta}^k$ will adjust the predicted circuit queue at the predecessor node.
This way, we can indirectly constrain the incoming rate by enforcing
\begin{equation}\label{eq:queue_cons_s_tilde}
	\tilde{\mathbf{s}}_{\alpha|\beta}^{k}\geq 0,
\end{equation}
which ensures that the incoming rate can only be increased as long as data is available in the source node.

A problem arises if the availability of data is not the rate-limiting factor at the source node.
To cope with this situation, we propose a new mechanism for the source node to request a rate increase,
which is incorporated as the second constraint on $\mathbf{r}_{\text{in}, \alpha}^k$.
For this mechanism, we introduce an additional state $\hat s_{\alpha}^k$ and input $\hat r_{\text{out}, \alpha}^k$
with the dynamic model equation:
\begin{align}
	\label{eq:mpc_model_s_hat}
	\hat s_{\alpha}^{k+1} &= \hat s_{\alpha,i}^k + \Delta t(r_{\text{in}, \alpha}^k - \hat r_{\text{out}, \alpha}^k).
\end{align}
State, input and dynamics are reminiscent of \eqref{eq:mpc_model}, with the difference that
$\hat s_{\alpha}^k$  denotes the virtual queue size and
$\hat r_{\text{out}, \alpha}^k$ the virtual outgoing rate.
These variables are virtual in the sense that
they \emph{do not} respect the rate constraint
imposed by the successor node in \eqref{eq:mpc_pred_r_out_max},
but only the local constraints:
\begin{equation}\label{eq:predictor_rate_larger_zero_r_out_hat}
	0 \leq \hat r_{\text{out},\alpha}^k,
\end{equation}
\begin{equation}\label{eq:predictor_link_capacity_constraints_r_out_hat}
	\sum_{i\in P_{\alpha}}  \hat r_{\text{out},\alpha}^k \leq C_{\alpha}^{\text{out}} ,
\end{equation}
and
\begin{equation}\label{eq:cons_queue_size_virtual}
	0\leq s_{\alpha}^k \leq \hat s_{\alpha}^{\text{max}}.
\end{equation}
In this regard, the virtual outgoing rate $\hat r_{\text{out}, \alpha,i}^k$
can be interpreted as the potential of node $\alpha$ to increase the rate for circuit $i$.
Node $\alpha$ receives the respective information from source node $\beta$
as $\mathbf{\hat r}_{\text{out},\beta}^{k}$.
The most straightforward way to consider this information at node $\alpha$ is to enforce:
\begin{equation}
	\label{eq:mpc_pred_r_in_max}
	{r}_{\text{in},\alpha}^{k} \leq
	{r}_{\text{in},\alpha}^{\text{max},k}
	={\hat r}_{\text{out},\beta}^{k},
\end{equation}
which complements \eqref{eq:mpc_pred_r_out_max} for the incoming rates.
In practice, however, we found that with
$
{r}_{\text{in},\alpha}^{\text{max},k}
={\hat r}_{\text{out},\beta}^{k}
$,
the incoming rate should be limited with:
\begin{equation}
	\label{eq:mpc_pred_r_in_max_sum}
	0\leq \sum_{l=0}^{k} \left[r_{\text{in},\alpha}^{\text{max},l} - {r}_{\text{in},\alpha}^{l} \right].
\end{equation}
In most cases, the effect of constraint \eqref{eq:mpc_pred_r_in_max_sum}
is similar to constraint~\eqref{eq:mpc_pred_r_in_max}.
The difference is that~\eqref{eq:mpc_pred_r_in_max}
enforces a concrete upper limit for the incoming rate at time $k$,
whereas~\eqref{eq:mpc_pred_r_in_max_sum} balances the limit over the horizon.
This is beneficial as the sequence $\mathbf{\hat r}_{\text{out},\beta}^{k}$
often contains single elements with very high data rates
which cannot necessarily be fully utilized by the successor at that exact point in time.

\subsection{Optimization Problem}\label{ssec:ocp}
In the following, we propose an optimal control problem for congestion control
with fairness formulation for node~$\alpha$, predecessor node~$\beta$ and successor node $\gamma$.
As optimization variables we introduce the states~$\mathbf{s}_\alpha$, $\Delta \mathbf{s}_{\alpha|\beta}$ and
$\mathbf{\hat s}_{\alpha}$ with their respective dynamics in~\eqref{eq:mpc_model}, \eqref{eq:mpc_update_eq_delta_s} and
\eqref{eq:mpc_model_s_hat}.
Furthermore, we optimize the inputs~$\Delta \mathbf{r}_{\text{in},\alpha}$ and $\Delta \mathbf{r}_{\text{out},\alpha}$ from which rates are determined with:
\begin{align}
	\label{eq:r_in_from_delta_r_in}
	\mathbf{r}_{\text{in},\alpha} &= \mathbf{r}^{\text{max}}-\Delta \mathbf{r}_{\text{in},\alpha}\\
	\label{eq:r_out_from_delta_r_out}
	\mathbf{r}_{\text{out},\alpha} &= \mathbf{r}^{\text{max}}-\Delta \mathbf{r}_{\text{in},\alpha}.
\end{align}
To express the newly introduced variable $\mathbf{\hat r}_{\text{out}, \alpha}$ we introduce
two optimization variables $\Delta \mathbf{r}_{\text{out,extra}}$ and $\mathbf{r}_{\text{out,minus}}$,
such that:
\begin{equation}
	\label{eq:r_hat_out_from_r_in_extra_r_in_minus}
	\mathbf{\hat r}_{\text{out}, \alpha} = \mathbf{r}_{\text{out}, \alpha} + \underbrace{(\mathbf{r}^{\text{max}}-\Delta \mathbf{\hat r}_{\text{out,extra}, \alpha})}_{\mathbf{\hat r}_{\text{out,extra}, \alpha}} - \mathbf{\hat r}_{\text{out,minus}, \alpha}.
\end{equation}
The virtual outgoing rate $\mathbf{\hat r}_{\text{out}, \alpha}$ is thus not chosen independently
but as a deviation from the outgoing rate $\mathbf{r}_{\text{out}, \alpha}$.
This mathematical construction was found to aid the convergence to the global max-min fair solution.

Considering~\eqref{eq:r_in_from_delta_r_in}, \eqref{eq:r_out_from_delta_r_out} and \eqref{eq:r_hat_out_from_r_in_extra_r_in_minus}, we state the OCP:
\begin{samepage}
\begin{small}
\begin{subequations}
\label{eq:mpc_full_optim}
\begin{equation}
\min_{
	\begin{array}{c}
		\mathbf{s}_{\alpha},
		\Delta \mathbf{s}_{\alpha|\beta},
		\mathbf{\hat s}_{\alpha},
		\Delta \mathbf{r}_{\text{out},\alpha},
		\Delta \mathbf{r}_{\text{in},\alpha},\\
		\Delta \mathbf{\hat r}_{\text{out,extra}},
		\mathbf{\hat r}_{\text{out,minus}}
	\end{array}
}\quad
\sum_{k=0}^{N_{\text{horz}}} d^k \left(
(\Delta r_{\text{in},\alpha}^k)^2+ (\Delta r_{\text{out},\alpha}^k)^2
+ (\Delta \hat{r}_{\text{out,extra},\alpha}^k)^2
+ ( \hat r_{\text{out,minus},\alpha}^k)^2
\right)
\tag{\ref{eq:mpc_full_optim}}
\end{equation}
\begin{align}
\text{\textbf{subject to}:}\quad & & \nonumber\\
\text{queue dynamics
\eqref{eq:mpc_model}, \eqref{eq:mpc_update_eq_delta_s}, \eqref{eq:mpc_model_s_hat}:
} \quad &s_{\alpha}^{k+1} = s_{\alpha}^k + \Delta t(r_{\text{in},\alpha}^k - r_{\text{out},  \alpha}^k), & \forall k = 0 ,\dots, N_{\text{horz}}\\
&\Delta s_{\alpha|\beta}^{k+1} = \Delta s_{\alpha|\beta}^{k} +\Delta t( r_{\text{in},\alpha}^k - r_{\text{out},\beta}^k) & \forall k = 0 ,\dots, N_{\text{horz}}\\
&\hat s_{\alpha}^{k+1} = \hat s_{\alpha}^k + \Delta t(r_{\text{in}, \alpha}^k - \hat r_{\text{out}, \alpha}^k) & \forall k = 0 ,\dots, N_{\text{horz}}\\
\text{queue constraints \eqref{eq:cons_queue_size}, \eqref{eq:queue_cons_s_tilde}, \eqref{eq:cons_queue_size_virtual}:
} \quad
&0 \leq s_{\alpha}^k \leq s_{\alpha}^{\text{max}}
\label{eq:mpc_full_optim_sc_max} & \forall k = 0 ,\dots, N_{\text{horz}}\\
&0 \leq \tilde{s}_{\alpha|\beta}^{k} & \forall k = 0 ,\dots, N_{\text{horz}}\\
&0 \leq \hat s_{\alpha}^k \leq s_{\alpha}^{\text{max}}
\label{eq:mpc_full_optim_shat_max} & \forall k = 0 ,\dots, N_{\text{horz}}\\
\text{rate in constraints \eqref{eq:predictor_r_in_positive}, \eqref{eq:predictor_link_capacity_constraints_in}, \eqref{eq:mpc_pred_r_in_max_sum}:
} \quad
&0\leq  r_{\text{in},\alpha}^k  & \forall k = 0 ,\dots, N_{\text{horz}}\\
&\sum_{i\in P_{\alpha}}  r_{\text{in},\alpha,i}^k  \leq C_{\alpha}^{\text{in}} & \forall k = 0 ,\dots, N_{\text{horz}}\\
&0\leq \sum_{l=0}^{k} \left[r_{\text{in},\alpha}^{\text{max},l} - {r}_{\text{in},\alpha}^{l} \right] & \forall k = 0 ,\dots, N_{\text{horz}}\\
\text{rate out constraints \eqref{eq:predictor_r_out_positive}, \eqref{eq:mpc_pred_r_out_max}, \eqref{eq:predictor_link_capacity_constraints_out}:
} \quad
&0\leq  r_{\text{out},\alpha}^k \leq r_{\text{out},\alpha}^{\text{max},k} & \forall k = 0 ,\dots, N_{\text{horz}}\\
&\sum_{i\in P_{\alpha}}   r_{\text{out},\alpha,i}^k \leq C_{\alpha}^{\text{out}} & \forall k = 0 ,\dots, N_{\text{horz}}\\
\text{virt. rate out constraints \eqref{eq:predictor_link_capacity_constraints_r_out_hat}, \eqref{eq:predictor_rate_larger_zero_r_out_hat}:
} \quad
&\sum_{i\in P_{\alpha}}  \hat r_{\text{out},\alpha,i}^k \leq C_{\alpha}^{\text{out}} &  \forall k = 0 ,\dots, N_{\text{horz}}\\
&0\leq  \hat r_{\text{out},\alpha}^k & \forall k = 0 ,\dots, N_{\text{horz}}\\
\text{initial conditions:
} \quad
&s_{\alpha}^0 = s_{\alpha}^{\text{init}},\ \Delta s_{\alpha|\beta}^{0} = 0,\ \hat s_{\alpha}=s_{\alpha}^{\text{init}} &
\end{align}
\end{subequations}
\end{small}
\end{samepage}%
The optimization problem~\eqref{eq:mpc_full_optim} is solved at each time-step and in each node of the network under consideration of the predicted trajectories of adjacent nodes~$\mathbf{r}_{\text{in},\gamma}^{k-1}$,
$\mathbf{r}_{\text{out},\beta}^k$ and $\mathbf{s}_{\beta}^k$,
$\mathbf{\hat r}_{\text{out},\beta}^{k}$
as well as the current size of the circuit queue in the current
node~$s_{\alpha}^{\text{init}}$.
Note that according to \eqref{eq:mpc_pred_r_out_max},
we set $\mathbf{r}_{\text{out},\alpha}^{\text{max},k}
=\mathbf{r}_{\text{in},\gamma}^{k-1}$
and according to \eqref{eq:mpc_pred_r_in_max}
we have $\mathbf{r}_{\text{in},\alpha}^{\text{max}}=\mathbf{\hat r}_{\text{out},\beta}$.

The objective in~\eqref{eq:mpc_full_optim} is motivated by the presented Theorem~\ref{theo:max_min_optim},
but with some important adaptations:
Most notably, we introduced~$\Delta r$ variables for both the incoming and outgoing rates.
Introducing the control variable~$\Delta r_{\text{in},\alpha}$ allows to control the incoming rate.
This  is of significant importance for the desired congestion control as it realizes \textit{backpressure} and data will be stopped from entering the network if it cannot be forwarded.
The quadratic term in~$\Delta r_{\text{out},\alpha}$ ensures that the circuit queue is emptied even if there are no new packets entering the node. 
Moreover, to further avoid congestion, PredicTor is \emph{explicitly} designed to limit the circuit queues through constraint~\eqref{eq:mpc_full_optim_sc_max}.

The objective in~\eqref{eq:mpc_full_optim} is further modified by introducing a discount factor~$d$.
This is necessary because naively implementing our presented fairness formulation also results in fairness along the prediction horizon,
where it is always preferable to increase the rate of the smallest element in a sequence for a given circuit.
In practice, however, we want to send and receive as soon as possible as long as \emph{instantaneous} fairness is achieved.
In the appendix of~\cite{PredicTor}, we present a guideline on how to choose~$d$ to obtain the desired behavior.

Note that the theoretical analysis of the feasibility and stability of the proposed MPC controller
in~\eqref{eq:mpc_full_optim} is out of the scope of this work.
We however found that the proposed controller is stable and feasible
for all investigated simulation studies.

\subsection{Controller Integration}\label{ssec:controller_tor_interact}

The proposed controller is implemented on the application layer of each node in the Tor network.
At each time step, problem~\eqref{eq:mpc_full_optim} is solved
with the most recent measurement of the circuit queue~$s_{\alpha}^{\text{init}}$
of the current node~$\alpha \in N$
and with the received information from adjacent nodes.
The optimal solution of~\eqref{eq:mpc_full_optim} is converted to trajectories
of incoming ($\mathbf{r}_{\text{in},\alpha}$) and outgoing ($\mathbf{r}_{\text{out},\alpha}$) rates,
where the first element of~$\mathbf{r}_{\text{out},\alpha}$
is used to control at which rate data is sent.
In particular, we employ a token-bucket method~\cite{bertsekas1992data}
to shape the outgoing traffic.
Note that we are controlling the data rates
on a per-circuit basis, which is similar to
how PCTCP~\cite{alsabah2013pctcp} changes Tor's circuit handling.

In order to exchange the trajectories between relays,
we extend the Tor protocol with respective control messages.
This is technically possible due to the extensibility of the Tor~protocol
that allows the definition of new cell types.
For the edges of each circuit that do not have a predecessor or successor to exchange data,
we provide reasonable, synthetic trajectories to bootstrap the data transfer
and behave accordingly.
For example, the first node in a circuit reads data from its source
according to its computed incoming rate.
While PredicTor is generally agnostic to the underlying transport protocol,
we implement it using TCP as a reliability mechanism,
to avoid packet loss and packet reordering.

\section{Security and Privacy Considerations} \label{sec:security}

Extending a widespread anonymity networks like Tor, in a fundamental way as PredicTor does,
requires great care to avoid reducing the security level
and putting the users' privacy to danger.
While this section does not constitute a formal, complete security analysis,
it aims to give an intuition of how PredicTor would interfere
with the existing security guarantees and privacy level in Tor.
In order to do so, we first recall Tor's adversary model
and afterwards review typical attacks on Tor, analyzing how PredicTor relates to them.
For this, we make use of the categorization of Tor attacks
presented in~\cite{DBLP:journals/csur/AlSabahG16}.
Still, PredicTor was primarily designed
for exploring new directions towards multi-hop congestion control
and does not claim to fully satisfy the security requirements needed for production use.

\paragraph{Adversary Model}

The Tor network aims to protect the privacy of its users.
More specifically, it obfuscates their IP addresses
by relaying the traffic over circuits of intermediate hops.
In this setting, Tor assumes the following adversary~\cite{dingledine2004tor}:
An attacker may control a certain fraction of the network, including underlay links and relays.
However, Tor does not protect against a global adversary that can monitor the whole network.
Consequently, Tor aims to protect against local adversaries,
but cannot currently avoid end-to-end traffic confirmation attacks.

One design aspect that goes hand-in-hand with this assumption
is that Tor does not rely on any specific trust relationship between relays.
In particular, the Tor protocol tries to make it as hard as possible
for cooperating malicious relays to de-anonymize users.
One essential building block for this is the use of telescope-like onion encryption,
such that each a relay within a circuit can only see its immediate predecessor and successor.
Payload data, including the target address,
is only visible end-to-end between the client and the exit.
All other information for building and maintaining circuits is only visible hop-by-hop,
so the insight every single relay can gain about the overall circuit is minimized.
The design of PredicTor's network protocol is in line with this general concept:
feedback data is only exchanged between adjacent relays,
utilizing Tor's existing cryptography and trust assumptions.

\paragraph{Information Leakage from Feedback Messages}

Although feedback messages themselves are only exchanged between directly neighboring relays,
one can argue that parts of the information they contain
actually trickle over more than one hop.
This is because feedback information is taken into account by the optimizer
and therefore implicitly influences the feedback information
which is, in turn, sent to other relays after the optimization step.
We cannot see how such information could be exploited by malicious relays,
but cannot prove this kind of information leakage irrelevant, either.
This would require a more in-depth analysis, which is subject to future work.

Our initial assessment, however, can be summarized as follows:
The potential danger would be that local relays
could learn more about the overall circuit than before,
\eg, inferring parts of the circuit topology.
For such inference, obtaining the exact number of circuits handled by another relay
may be useful to the adversary.
Feedback messages in PredicTor contain per-circuit trajectories,
but only for the circuits that are multiplexed between any two adjacent relays.
The same piece of information is already available to vanilla Tor relays for circuit handling.
Obtaining the overall number of circuits handled by another relay would, just as before,
require heuristics that make use of additional information,
such as the data rate advertised in the network consensus.
We therefore do not think that PredicTor would facilitate such inference.
The same is true for the exchanged data rates.
Since these are propagated between relays after optimization,
malicious relays could try to infer characteristics of relays at the far end of the circuit,
\eg, ruling out relay candidates that would lead to different bottleneck values.
However, such inference could also be carried out as of today,
simply by locally measuring the throughput achieved per circuit.

\paragraph{Traffic Confirmation Attacks}

The most fundamental vulnerability of Tor is
its susceptibility to timing-based traffic confirmation attacks
that could either be carried out by a global adversary or by an adversary
that controls, both, the entry and the exit relay.
In general, low-latency anonymity networks cannot prevent traffic confirmation attacks.
However, even if such attacks cannot be prevented entirely,
their difficulty depends largely on the temporal correlation
between data entering and leaving the network.
Consequently, better and more predictable performance, as is achieved by PredicTor,
may facilitate traffic confirmation attacks.
This is an inherent challenge every performance enhancement for Tor has to face.
On the other hand, better performance bears the potential to attract a broader user base,
strengthening the anonymity set,
so it is not entirely clear whether the overall privacy level would be harmed.
The precise interdependencies between these factors are one of our main future research directions.

\paragraph{Routing and Circuit Selection}

Adversaries may be able to increase their chances
of successfully attacking a larger portion of the user base
by taking into account the circuit selection process that is carried out by the Tor clients.
Likewise, if adversaries not only control individual relays, but Autonomous Systems~(AS),
this can be problematic.
PredicTor, however, is fully orthogonal and does not touch circuit selection or routing.
Therefore, it does not increase vulnerability to these attacks.

\paragraph{Website Fingerprinting and Watermarking}

One attack vector that has extensively been researched in the past
works by inferring communication contents by analyzing the timing of the encrypted data traffic.
Closely related, there are watermarking techniques
that actively insert traffic pattern peculiarities
in order to make flows more easily recognizable.
We argue that both attack strategies are either not touched by PredicTor
or even become more difficult because the explicit traffic handling defined by PredicTor
results in smoother and more homogeneous traffic patterns on the wire,
as we will show in Section~\ref{sec:eval-small-network}.

\paragraph{Congestion Attacks and Denial of Service~(DoS)}

Attackers can try to divert target traffic from relays they do not control
to their own ones by artificially congesting parts of the network
up to the point where benign relays cannot continue operation~(DoS).
As far as traffic congestion is concerned,
PredicTor is likely more resilient to such attacks than vanilla Tor,
because it handles situations of congestion or load much more explicitly.
Attackers cannot as easily flood the network,
because the controllers running at each relay optimize the traffic flows
to keep queues low.
On the other hand, we suspect PredicTor to be prone to DoS~attacks
that specifically target the optimizer.
Although the underlying optimization problem can be solved efficiently,
the computational effort still depends on the number of circuits involved.
Therefore, an attacker could trigger increased resource usage
by constructing tailored circuits that increase the difficulty of PredicTor's optimization problem.

\vspace{\baselineskip}
All in all, our initial security and privacy assessment
shows that future work would be necessary to ensure
that PredicTor does not introduce new attack vectors.
On the other hand, it can also help mitigating several existing attacks.

\section{Evaluation and Discussion} \label{sec:eval}

Evaluating PredicTor on the live Tor network is prohibitive,
due to its sensitive nature of touching users' anonymity.
Instead, we here investigate its behavior by carrying out simulation studies.

It should be noted that, in its current form,
PredicTor is not meant for immediate real-world application on the Tor network.
Instead, it constitutes a concept study opening a novel research direction
towards realizing congestion control in complex networks.
Therefore, all of our experimentation has the goal
of maximizing the \emph{understanding} of this new approach,
applying distributed model predictive control for congestion control.
We aim to clearly point out benefits and potential drawbacks of such strategies.
In particular, our evaluation covers the PredicTor controller's detailed behavior and its implications.
We look at the isolated behavior of single controllers as well as
the overall system behavior that emerges from the cooperative interaction
of multiple controllers.
To put the results into context, we compare PredicTor to vanilla Tor
as well as PCTCP~\cite{alsabah2013pctcp}, an alternative circuit handling strategy for Tor.
From these observations, we deduce insight about the benefits, inherent limitations,
as well as the expected applicability of such approaches.

Our evaluation strategy is twofold: First, we analyze PredicTor's behavior
in small toy scenarios that allow us to better understand its behavior in situations
that are simple enough to be investigated by hand.
By doing so, we establish an intuition of its behavior.
To this end, we first investigate the working of a single, isolated PredicTor controller
before setting up multiple controllers to cooperate.
This way, we can analyze the behavior of a single, isolated PredicTor controller
as well as the interaction between multiple controllers.

At a second step, we take our evaluation further by scaling up our experiments to more complex networks.
This serves as a means of exploring the applicability of PredicTor in more realistic scenarios.
On the one hand, we thus verify whether PredicTor's claimed benefits
do also exist in scenarios that are not as easy to understand as the toy scenarios before.
On the other hand, we investigate to what extent the underlying assumptions made in PredicTor
collide with reality and what implications this has for further research in the field.
Lastly, we evaluate how PredicTor handles different traffic patterns (bulk and web traffic).

\paragraph{Implementation} \label{sec:implementation}

As laid out before, PredicTor comprises, on the one hand, the controller logic
that uses model predictive control to find the optimal data rates
based on the implemented system model and optimization objectives.
On the other hand, these control decisions have to be realized in the network
and the predicted trajectories have to be exchanged with other relays.
The structure of our prototype implementation of PredicTor also exhibits these two main tasks.
We implement the PredicTor core model in Python, utilizing CasADi~\cite{Andersson2018}
in combination with IPOPT~\cite{Andreas2006} and the MA27\footnote{HSL.
A collection of Fortran codes for large scale scientific computation. \texttt{http://www.hsl.rl.ac.uk/}}
linear solver for fast state-of-the-art optimization.
For implementing the network behavior, we embed it into \texttt{nstor}~\cite{tschorsch16bktap},
an implementation of Tor for the ns-3~network simulator.\footnote{\texttt{https://www.nsnam.org/}}
Our implementation thus covers PredicTor, vanilla Tor and PCTCP.
PCTCP differs from vanilla~Tor in that it establishes a separate connection per circuit,
instead of multiplexing them.
While the overall network simulation is carried out by \texttt{nstor},
each simulated relay has access to the controller code library for carrying out its local optimizations.
The results are then used as the Tor relay's scheduling strategy in the network simulation,
More specifically, this means throttling data transmission
based on the controller's optimization output and exchanging the predicted trajectories between relays.
Our implementation is publicly available online.\footnote{\texttt{https://github.com/cdoepmann/PredicTor}}

The results presented in the following are obtained
with a discount factor of $d = \frac{1}{3}$,
as discussed in the appendix of~\cite{PredicTor}.
For the prediction horizon, we choose $N_{\text{horz}}=10$

\paragraph{Metrics} \label{sec:metrics}

In order to assess PredicTor's performance,
we quantify different parameters that are relevant for comparing it to existing approaches.
For each of these parameters, we initially evaluate the steady state behavior.
First of all, we consider the \emph{latency} of the data transfer.
Apart from the physical underlay latency, which denotes a natural lower bound,
latency stems primarily from the existence of buffers and queues in the network.
Since the reduction of queue sizes is an explicit optimization goal of PredicTor,
we expect a considerable enhancement in this regard.
We define our notion of latency as follows:
For each data transfer through the network (running over a circuit of multiple relays),
we define latency as the difference in time between
when it reaches its destination and when it entered the (overlay) network.
We therefore focus explicitly on the Tor network itself
and disregard additional latency that may occur, \eg,
for the communication between the exit relay and an outside webserver.
This approach ensures that we capture the two important consequences of latency:
Firstly, the impact on the user who experiences additional waiting time
before seeing the response to her request.
And secondly, large queue sizes also mean a higher load on the network itself.
Therefore, small queues---and thus, lower latency---are desirable
also from a network perspective in reducing congestion.
We refer to the total amount of data that is present in the network at any point in time as \emph{backlog}.
For latency, we employ a byte-wise perspective.
That is, we precisely track for each payload byte the times of sending and receiving
and aggregate them into an overall value by taking the average.

Another relevant metric is the \emph{throughput} that is achieved by each of the circuits.
On the one hand, it expresses how well the network is utilized.
On the other hand, we can use these characteristics to define a notion of \emph{fairness} in the network.
As explained before, PredicTor aims at establishing max-min fairness for all circuits,
which Tor is not currently capable of.
Given a specific topology of circuits and relays,
we calculate the optimal max-min fair data rate distribution for this scenario.
We can then compare the observed values from our simulations to this optimum in two ways:
For an in-depth analysis of the resulting data rate distributions,
we visualize them as cumulative distribution functions in a CDF plot.
On the other hand, however, if we only want a single value to express the degree of fairness,
we make use of the following construction:
Let $r^f_1, \dots, r^f_n$ be the max-min fair data rate distribution,
and $r_1, \dots, r_n$ be the observed data rates.
We then define the fairness index $F$ as follows:
\[
F = 1 - \frac{\sum_{i=1}^n \vert r^f_i - r_i \vert}{\sum_{i=1}^n r^f_i}
\]
Put differently, $F$ denotes the share of traffic that behaves according to max-min fairness.
While other established fairness measures exist,
most notably Jain's fairness index~\cite{jains-fairness-index},
they are not applicable here.
Jain defines fairness as a uniform allocation of a resource.
Since we make use of max-min fairness, however,
a uniform data rate distribution is not necessarily the optimal choice.

\paragraph{Limitations} \label{sec:limitations}

One of our main simplifications includes that
we do not simulate the transmission of feedback messages over the wire,
but emulate their exchange \enquote{out of band}.
Since feedback traffic is independent of the network speed,
running the simulations at different simulated data rates,
one could achieve arbitrary goodput-overhead ratios, which we refrain from doing.
Moreover, our intention is to evaluate the scheduling behavior itself as a baseline
instead of capturing artifacts that stem purely from implementation details such as packet format.
However, we note (and later discuss) that, in real networks,
feedback overhead would constitute a severe issue due to its linear growth with the number of circuits.
Moreover, for this prototype, we base the data transfer on TCP (like Tor and PCTCP)
instead of introducing a tailored transport protocol.
PredicTor thus merely takes the role of a scheduler.
We will later discuss that this approach can still be beneficial for robustness.

\subsection{Single Controller}

We first present an investigation of the decision-making process of PredicTor's
proposed controller from~\eqref{eq:mpc_full_optim} for a single node.
In order to highlight several interesting aspects of its behavior,
we investigate an \emph{open-loop prediction}.
This means that we solve the optimal control problem~\eqref{eq:mpc_full_optim} once
and display the resulting predicted future trajectories.
In the \emph{closed-loop control} application, these predictions will change repeatedly, as new information becomes available.

For the investigation, we consider the small-scale topology presented in Figure~\ref{fig:example-topology}.
The topology consists of six relays handling a total of three circuits.
All of the three circuits meet in a shared bottleneck.
Additionally, two of these circuits originate from the same sending relay.
The scenario therefore demonstrates a simple congestion situation.
We focus on the controller at the bottleneck and
denote $\alpha$ the current node, $\beta$ its predecessors and $\gamma$ its successors.

\begin{figure}
    \definecolor{color0}{rgb}{0.12156862745098,0.466666666666667,0.705882352941177}
    \definecolor{color1}{rgb}{1,0.498039215686275,0.0549019607843137}
    \definecolor{color2}{rgb}{0.172549019607843,0.627450980392157,0.172549019607843}
    \begin{tikzpicture}[font=\scriptsize]
      \node (v) {circuit 1};
      \draw[color=color0,very thick] ($(v.east) + (0.15cm,0cm)$) -- +(0.7cm,0cm);

      \node[right=1.75cm of v] (s) {circuit 2};
      \draw[color=color1,very thick] ($(s.east) + (0.15cm,0cm)$) -- +(0.7cm,0cm);

      \node[right=1.75cm of s] (n) {circuit 3};
      \draw[color=color2,very thick] ($(n.east) + (0.15cm,0cm)$) -- +(0.7cm,0cm);
    \end{tikzpicture} \par
    \tikzset{
    circuit/.style={line width=3.5pt},
    written/.style={text opacity=0},
    }
    \def\xfactor{0.6}
    \hfill
    \scalebox{0.4}{
    \begin{tikzpicture}[]
        \drawtopology
        \node[draw,fit=(sender 1) (sender 2),dashed,ultra thick] {};
    \end{tikzpicture}
    }
    \hfill
    \scalebox{0.4}{
    \begin{tikzpicture}[]
        \drawtopology
        \coordinate (incoming) at ($(sender 1.center)!0.5!(bottleneck.center)$);
        \draw[dashed,ultra thick] ($(incoming) + (-1.5em,+1.5em)$) rectangle ($(incoming) + (+1.5em,-4.5em)$);
    \end{tikzpicture}
    }
    \hfill
    \scalebox{0.4}{
    \begin{tikzpicture}[]
        \drawtopology
        \node[draw,fit=(bottleneck),dashed,ultra thick] {};
    \end{tikzpicture}
    }
    \hfill
    \scalebox{0.4}{
    \begin{tikzpicture}[]
        \drawtopology
        \coordinate (outgoing) at ($(receiver 2.center)!0.5!(bottleneck.center)$);
        \draw[dashed,ultra thick] ($(outgoing) + (-1.5em,+4.0em)$) rectangle ($(outgoing) + (+1.5em,-4.5em)$);
    \end{tikzpicture}
    }
    \hfill
    \par\vspace{0.5em}
	\pgfplotsset{width=0.29\textwidth, height=0.35\textwidth}
	\pgfplotsset{
		tick label style={font=\scriptsize\sffamily},
		label style={font=\scriptsize\sffamily},
		y label style={yshift = -0.01\textwidth},
		title style={font=\scriptsize\bfseries\sffamily},
		legend style={font=\tiny},
		group/horizontal sep=1.1cm,
		group/vertical sep = 0.8cm,
	}
	\input{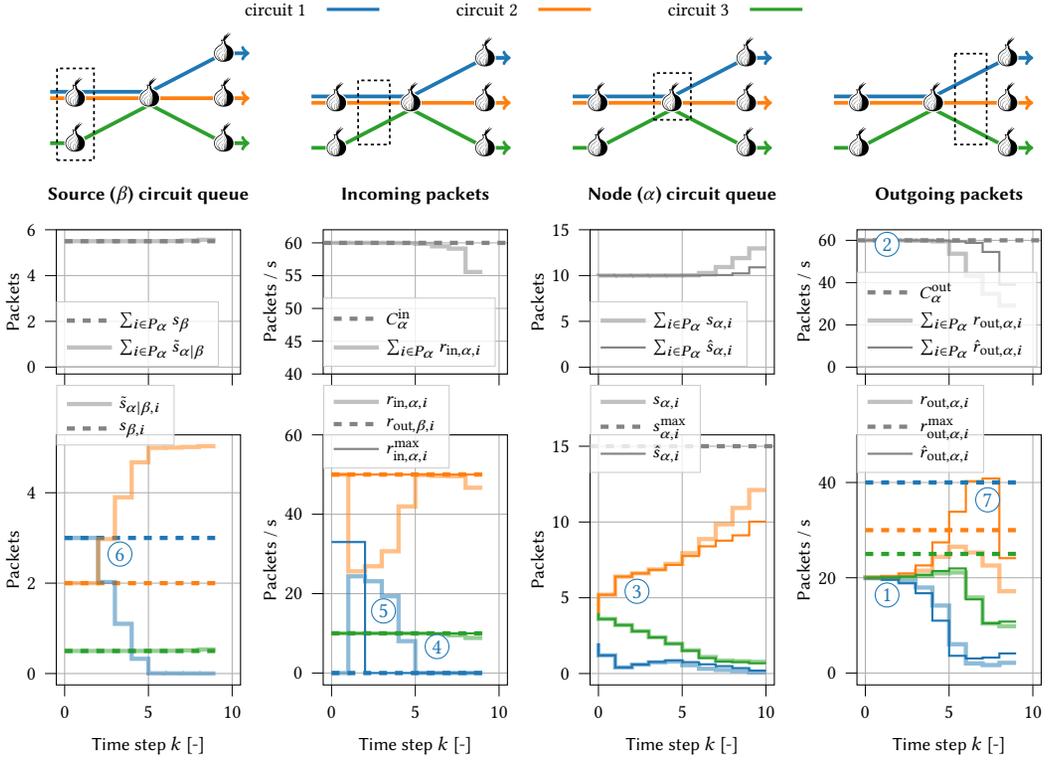}
	\caption{Open-loop simulation of the MPC predictions
		made by the controller in the bottleneck relay of our sample topology.}
	\label{fig:open_loop_pred}
\end{figure}
We investigate a synthetic scenario
with manually-defined trajectories from the adjacent nodes.
In particular, we define for the predecessor node
$\mathbf{r}_{\text{out},\beta}^k$,
$\mathbf{\hat r}_{\text{out},\beta}^k$,
and $\mathbf{s}_{\beta}^k$.
For the successor node we define
$\mathbf{r}_{\text{in},\gamma}^{k-1}$.
We set $\mathbf{r}_{\text{out},\alpha}^{\text{max},k}
=\mathbf{r}_{\text{in},\gamma}^{k-1}$
and $\mathbf{r}_{\text{in},\alpha}^{\text{max},k}=\mathbf{\hat r}_{\text{out},\beta}^k$.
Additionally, the initial buffer size $s_{\alpha}^{\text{init},k}$  is required.
Based on this information, the solution of \eqref{eq:mpc_full_optim} allows to compute the trajectories
$\mathbf{r}_{\text{in},\alpha}^k$,
$\mathbf{r}_{\text{out},\alpha}^k$,
$\mathbf{s}_{\alpha,i}^k$,
$\mathbf{\hat{s}}_{\alpha}^k$ and
$\mathbf{\tilde{s}}_{\alpha|\beta}^k$.
The trajectories are displayed in Figure~\ref{fig:open_loop_pred} and will be discussed in the following.

The most important decision of the PredicTor controller is with respect to the outgoing rates~$\mathbf{r}_{\text{out},\alpha}$.
The first element of this sequence determines the rate at which data is sent until the next sampling instance.
In Figure~\ref{fig:open_loop_pred},
we can see at~\circled{1} that all circuits obtain the same rate at the beginning of the sequence.
At this point in time the rate is only limited by the node capacity~$C_{\alpha}^{\text{out}}$, which can be seen in~\circled{2}.
Over the course of the prediction horizon, the outgoing rate for circuit $2$ is increasing,
whereas the rates for circuits~1 and 3 is decreasing.
This behavior is due to the buffer sizes for the circuits, shown in~\circled{3}.
Here we can see that even with an increasing rate,
the buffer for circuit~2 is growing whereas that for circuits~1 and 3 is approaching zero.
With the buffer size for circuits~1 and 3 close to zero, the outgoing rate for these circuits
approaches the incoming rates~$\mathbf{r}_{\text{in},\alpha}$, shown in~\circled{4},
at the end of the horizon.
Under stationary conditions, the incoming rates~$\mathbf{r}_{\text{in},\alpha}$ are typically equivalent
to the predicted outgoing rates~$\mathbf{r}_{\text{out},\beta}$
of the source node $\beta$, which can also be seen in~\circled{4} for circuit~3.
For circuit~1, however, the algorithm determines to
increase
$\mathbf{r}_{\text{in},\alpha,1}$
with respect to
$\mathbf{r}_{\text{out},\beta,1}$.
This can be seen in~\circled{5}.
When increasing the incoming rate, PredicTor needs to consider
$\mathbf{r}_{\text{in},\alpha}^{\text{max}}$
and the respective constraint in~\eqref{eq:mpc_pred_r_in_max_sum}.
Note that this constraint is not enforced at each time-step with
$\mathbf{r}_{\text{in},\alpha}\leq \mathbf{r}_{\text{in},\alpha}^{\text{max}}$
but the limit increases at each time-step with
$\mathbf{r}_{\text{in},\alpha}^{\text{max}}\geq 0$.
This can also be seen for circuit~1 in~\circled{5}.

Since the algorithm determines to increase the incoming rate
$\mathbf{r}_{\text{in},\alpha,1}$
with respect to the prediction of the source node
$\mathbf{r}_{\text{out},\beta,1}$,
it also predicts a deviation in the buffer size at the source node.
In \circled{6}, we can see that originally the buffer size $\mathbf{s}_{\beta,1}$ for circuit 1
is predicted to be constant over the horizon.
The obtained trajectory $\tilde{s}_{\alpha|\beta,1}$ takes into consideration
the increased $\mathbf{r}_{\text{in},\alpha,1}$ and predicts
that the buffer for circuit 1 at node $\beta$ will be emptied at time-step $k=5$.
Consequently, we can see at \circled{4} that the rate $r_{\text{in},\alpha,1}^{k+5}$ is zero.

As a final aspect, we want to mention $\mathbf{\hat r}_{\text{out},\alpha}$.
This virtual outgoing rate is different to $\mathbf{ r}_{\text{out},\alpha}$ as it does not consider the
constraint $\mathbf{r}_{\text{out},\alpha}\leq \mathbf{ r}_{\text{out},\alpha}^{\text{max}}$.
This can be seen in \circled{7} for circuit 2.
The purpose of this virtual rate is to request from the successor node an increase in
$\mathbf{ r}_{\text{out},\alpha}^{\text{max}}$, such that at a later iteration the true rate can be increased.

In summary, we find that the solution of the PredicTor problem
(see~\eqref{eq:mpc_full_optim})
leads to sound and interpretable behavior in terms of its predicted future trajectories.
For the synthetic scenario, it can be seen that the controller attempts to achieve fairness and to avoid congestion,
while utilizing the available resources of the network.
The results allow no conclusions regarding performance, however,
as only a single controller in an open-loop solution is investigated.

\subsection{Interaction of Multiple Controllers} \label{sec:eval-small-network}

As a next step, we carried out a full simulation of the sample network.
This now includes not only the isolated controllers,
but also the interaction between them as well as the application logic and network stack behavior.
To this end, ns-3 allows us to achieve a high degree of realism
by emulating the network down to the physical layer,
including queuing effects, packet loss, etc.
We refer to this kind of simulation as \emph{closed-loop},
because each MPC step takes into account the current state of the system,
determined from measurements and information exchange between relays.
For this simple setup, we define all underlay links to have the same, constant latency.
We will lift this assumption in later experiments.
We compare the performance of PredicTor to vanilla Tor and PCTCP.

We investigate two scenarios that differ slightly by circuit behavior:
The three circuits start at slightly different times, purely for easier visualization.
In scenario~1, circuits~1 and~3 have an infinite source of packets to forward,
whereas circuit~2 stops and restarts twice during the simulation.
This allows us to better investigate how PredicTor assigns data rates to each of the circuits.
In scenario~2, all circuits have an infinite source of packets.
We use this setup for comparing the absolute values of achieved data rates more easily.

\definecolor{color0}{rgb}{0.12156862745098,0.466666666666667,0.705882352941177}
\definecolor{color1}{rgb}{1,0.498039215686275,0.0549019607843137}
\definecolor{color2}{rgb}{0.172549019607843,0.627450980392157,0.172549019607843}

\begin{figure}
    \pgfplotsset{width=.38\textwidth,height=.275\textwidth}
    \pgfplotsset{
    tick label style={font=\scriptsize\sffamily},
    label style={font=\scriptsize\sffamily},
    title style={font=\scriptsize\bfseries\sffamily},
    legend style={font=\footnotesize},
    group/horizontal sep=0.5cm,
    }
    \begin{tikzpicture}[font=\scriptsize]
    \node (v) {circuit 1};
    \draw[color=color0,very thick] ($(v.east) + (0.15cm,0cm)$) -- +(0.7cm,0cm);

    \node[right=1.75cm of v] (s) {circuit 2};
    \draw[color=color1,very thick] ($(s.east) + (0.15cm,0cm)$) -- +(0.7cm,0cm);

    \node[right=1.75cm of s] (n) {circuit 3};
    \draw[color=color2,very thick] ($(n.east) + (0.15cm,0cm)$) -- +(0.7cm,0cm);
    \end{tikzpicture} \par

    \input{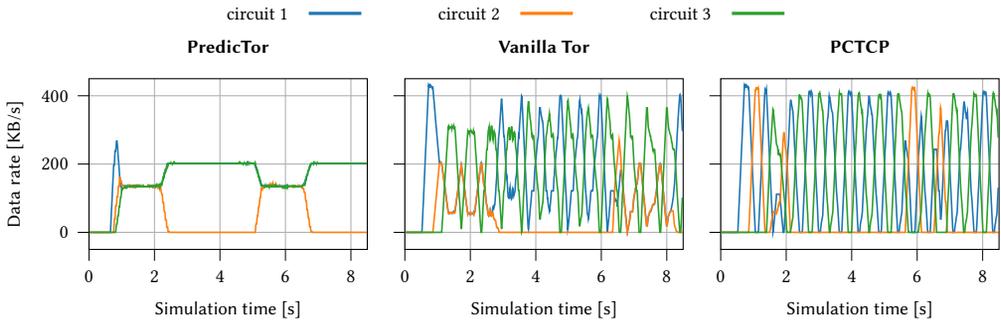}
	\caption{Data rates of the circuits in the sample topology (scenario~1) over time.}
	\label{fig:multiple-controllers-data-rate}
\end{figure}

Figure~\ref{fig:multiple-controllers-data-rate} shows the data rates of the individual circuits
in scenario~1 over the course of the simulation time.
These per-circuit values were measured in ns-3 by recording the outgoing rates at the bottleneck.
We can see that PredicTor exhibits a desirable behavior with constant,
sustainable rates and smooth transitions when circuit~2 stops and restarts.
Fair behavior can be observed in these transitions:
All circuits share the same rate during activity
and circuits~1 and~3 are allocated the same, higher rate when circuit~2 stops sending.
The sum of all rates is visibly constant over time.
On the other hand, vanilla Tor and PCTCP show erratic, oscillatory behavior
where bursts are followed by very low rates,
while the individual circuits take turns sending data.
Fairness cannot be assessed visually for Tor and PCTCP,
which is why we quantify it later. %

\begin{figure}
    \pgfplotsset{width=.38\textwidth,height=.275\textwidth}
    \pgfplotsset{
    tick label style={font=\scriptsize\sffamily},
    label style={font=\scriptsize\sffamily},
    title style={font=\scriptsize\bfseries\sffamily},
    legend style={font=\footnotesize},
    group/horizontal sep=0.5cm,
    }
    \begin{tikzpicture}[font=\scriptsize]
    \node (v) {circuit 1};
    \draw[color=color0,very thick] ($(v.east) + (0.15cm,0cm)$) -- +(0.7cm,0cm);

    \node[right=1.75cm of v] (s) {circuit 2};
    \draw[color=color1,very thick] ($(s.east) + (0.15cm,0cm)$) -- +(0.7cm,0cm);

    \node[right=1.75cm of s] (n) {circuit 3};
    \draw[color=color2,very thick] ($(n.east) + (0.15cm,0cm)$) -- +(0.7cm,0cm);
    \end{tikzpicture} \par

    \input{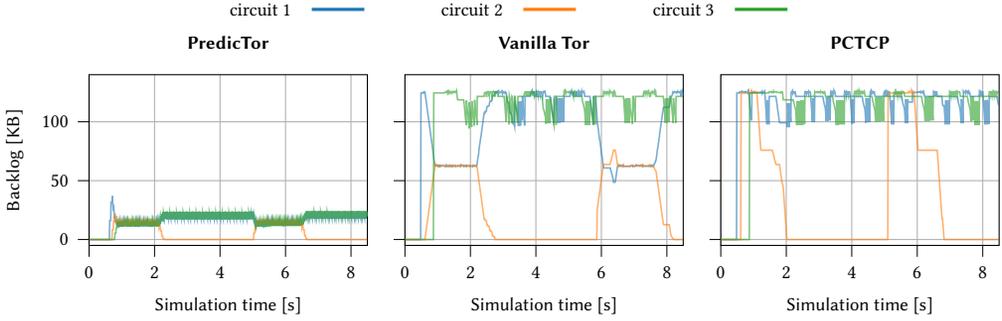}
    \vspace{-0.5em}
	\caption{Backlog of the circuits in the sample topology (scenario~1) over time.}
	\label{fig:multiple-controllers-backlog}
\end{figure}
\begin{figure}
    \pgfplotsset{width=.38\textwidth,height=.275\textwidth}
    \pgfplotsset{
    tick label style={font=\scriptsize\sffamily},
    label style={font=\scriptsize\sffamily},
    title style={font=\scriptsize\bfseries\sffamily},
    legend style={font=\footnotesize},
    group/horizontal sep=0.5cm,
    }
    \begin{tikzpicture}

\definecolor{color0}{rgb}{0.12156862745098,0.466666666666667,0.705882352941177}

\begin{groupplot}[group style={group size=3 by 1}]
\nextgroupplot[
tick align=outside,
tick pos=left,
title={PredicTor},
x grid style={white!69.0196078431373!black},
xlabel={Latency [ms]},
xmajorgrids,
xmin=0, xmax=1000,
xtick style={color=black},
y grid style={white!69.0196078431373!black},
ylabel={Amount of data [KB]},
ymajorgrids,
ymin=0, ymax=3000,
ytick style={color=black}
]
\draw[draw=none,fill=color0] (axis cs:60,0) rectangle (axis cs:89.5,259.5);
\draw[draw=none,fill=color0] (axis cs:90,0) rectangle (axis cs:119.5,2352);
\draw[draw=none,fill=color0] (axis cs:120,0) rectangle (axis cs:149.5,609);
\draw[draw=none,fill=color0] (axis cs:150,0) rectangle (axis cs:179.5,12.5);

\nextgroupplot[
scaled y ticks=manual:{}{\pgfmathparse{#1}},
tick align=outside,
tick pos=left,
title={Vanilla Tor},
x grid style={white!69.0196078431373!black},
xlabel={Latency [ms]},
xmajorgrids,
xmin=0, xmax=1000,
xtick style={color=black},
y grid style={white!69.0196078431373!black},
ymajorgrids,
ymin=0, ymax=3000,
ytick style={color=black},
yticklabels={}
]
\draw[draw=none,fill=color0] (axis cs:60,0) rectangle (axis cs:89.5,3);
\draw[draw=none,fill=color0] (axis cs:90,0) rectangle (axis cs:119.5,12.5);
\draw[draw=none,fill=color0] (axis cs:120,0) rectangle (axis cs:149.5,12.5);
\draw[draw=none,fill=color0] (axis cs:150,0) rectangle (axis cs:179.5,12.5);
\draw[draw=none,fill=color0] (axis cs:180,0) rectangle (axis cs:209.5,12.5);
\draw[draw=none,fill=color0] (axis cs:210,0) rectangle (axis cs:239.5,13);
\draw[draw=none,fill=color0] (axis cs:240,0) rectangle (axis cs:269.5,12.5);
\draw[draw=none,fill=color0] (axis cs:270,0) rectangle (axis cs:299.5,13);
\draw[draw=none,fill=color0] (axis cs:300,0) rectangle (axis cs:329.5,151.5);
\draw[draw=none,fill=color0] (axis cs:330,0) rectangle (axis cs:359.5,26.5);
\draw[draw=none,fill=color0] (axis cs:360,0) rectangle (axis cs:389.5,25.5);
\draw[draw=none,fill=color0] (axis cs:390,0) rectangle (axis cs:419.5,14);
\draw[draw=none,fill=color0] (axis cs:420,0) rectangle (axis cs:449.5,13.5);
\draw[draw=none,fill=color0] (axis cs:450,0) rectangle (axis cs:479.5,14);
\draw[draw=none,fill=color0] (axis cs:480,0) rectangle (axis cs:509.5,46.5);
\draw[draw=none,fill=color0] (axis cs:510,0) rectangle (axis cs:539.5,83.5);
\draw[draw=none,fill=color0] (axis cs:540,0) rectangle (axis cs:569.5,395);
\draw[draw=none,fill=color0] (axis cs:570,0) rectangle (axis cs:599.5,495.5);
\draw[draw=none,fill=color0] (axis cs:600,0) rectangle (axis cs:629.5,1895);
\draw[draw=none,fill=color0] (axis cs:630,0) rectangle (axis cs:659.5,17);
\draw[draw=none,fill=color0] (axis cs:750,0) rectangle (axis cs:779.5,0.5);
\draw[draw=none,fill=color0] (axis cs:780,0) rectangle (axis cs:809.5,0.5);
\draw[draw=none,fill=color0] (axis cs:660,0) rectangle (axis cs:689.5,9);
\draw[draw=none,fill=color0] (axis cs:690,0) rectangle (axis cs:719.5,9);
\draw[draw=none,fill=color0] (axis cs:720,0) rectangle (axis cs:749.5,2.5);

\nextgroupplot[
scaled y ticks=manual:{}{\pgfmathparse{#1}},
tick align=outside,
tick pos=left,
title={PCTCP},
x grid style={white!69.0196078431373!black},
xlabel={Latency [ms]},
xmajorgrids,
xmin=0, xmax=1000,
xtick style={color=black},
y grid style={white!69.0196078431373!black},
ymajorgrids,
ymin=0, ymax=3000,
ytick style={color=black},
yticklabels={}
]
\draw[draw=none,fill=color0] (axis cs:60,0) rectangle (axis cs:89.5,3);
\draw[draw=none,fill=color0] (axis cs:90,0) rectangle (axis cs:119.5,12.5);
\draw[draw=none,fill=color0] (axis cs:120,0) rectangle (axis cs:149.5,12.5);
\draw[draw=none,fill=color0] (axis cs:150,0) rectangle (axis cs:179.5,12.5);
\draw[draw=none,fill=color0] (axis cs:180,0) rectangle (axis cs:209.5,12.5);
\draw[draw=none,fill=color0] (axis cs:210,0) rectangle (axis cs:239.5,13);
\draw[draw=none,fill=color0] (axis cs:240,0) rectangle (axis cs:269.5,12.5);
\draw[draw=none,fill=color0] (axis cs:270,0) rectangle (axis cs:299.5,13);
\draw[draw=none,fill=color0] (axis cs:300,0) rectangle (axis cs:329.5,40.5);
\draw[draw=none,fill=color0] (axis cs:330,0) rectangle (axis cs:359.5,25);
\draw[draw=none,fill=color0] (axis cs:360,0) rectangle (axis cs:389.5,24);
\draw[draw=none,fill=color0] (axis cs:600,0) rectangle (axis cs:629.5,886);
\draw[draw=none,fill=color0] (axis cs:720,0) rectangle (axis cs:749.5,86);
\draw[draw=none,fill=color0] (axis cs:750,0) rectangle (axis cs:779.5,140);
\draw[draw=none,fill=color0] (axis cs:780,0) rectangle (axis cs:809.5,293.5);
\draw[draw=none,fill=color0] (axis cs:810,0) rectangle (axis cs:839.5,30.5);
\draw[draw=none,fill=color0] (axis cs:840,0) rectangle (axis cs:869.5,67.5);
\draw[draw=none,fill=color0] (axis cs:1050,0) rectangle (axis cs:1079.5,1);
\draw[draw=none,fill=color0] (axis cs:510,0) rectangle (axis cs:539.5,17);
\draw[draw=none,fill=color0] (axis cs:540,0) rectangle (axis cs:569.5,419);
\draw[draw=none,fill=color0] (axis cs:570,0) rectangle (axis cs:599.5,601);
\draw[draw=none,fill=color0] (axis cs:630,0) rectangle (axis cs:659.5,29.5);
\draw[draw=none,fill=color0] (axis cs:900,0) rectangle (axis cs:929.5,229.5);
\draw[draw=none,fill=color0] (axis cs:870,0) rectangle (axis cs:899.5,149.5);
\draw[draw=none,fill=color0] (axis cs:930,0) rectangle (axis cs:959.5,7.5);
\draw[draw=none,fill=color0] (axis cs:960,0) rectangle (axis cs:989.5,7.5);
\draw[draw=none,fill=color0] (axis cs:390,0) rectangle (axis cs:419.5,13);
\draw[draw=none,fill=color0] (axis cs:420,0) rectangle (axis cs:449.5,12.5);
\draw[draw=none,fill=color0] (axis cs:450,0) rectangle (axis cs:479.5,12.5);
\draw[draw=none,fill=color0] (axis cs:480,0) rectangle (axis cs:509.5,12.5);
\draw[draw=none,fill=color0] (axis cs:660,0) rectangle (axis cs:689.5,39.5);
\draw[draw=none,fill=color0] (axis cs:690,0) rectangle (axis cs:719.5,41);
\draw[draw=none,fill=color0] (axis cs:990,0) rectangle (axis cs:1019.5,9);
\draw[draw=none,fill=color0] (axis cs:1020,0) rectangle (axis cs:1049.5,3);
\end{groupplot}

\end{tikzpicture}
    \vspace{-1.5em}
	\caption{Histogram of per-byte latency in the sample topology (scenario~1) for all three circuits.}
	\label{fig:multiple-controllers-latency}
\end{figure}
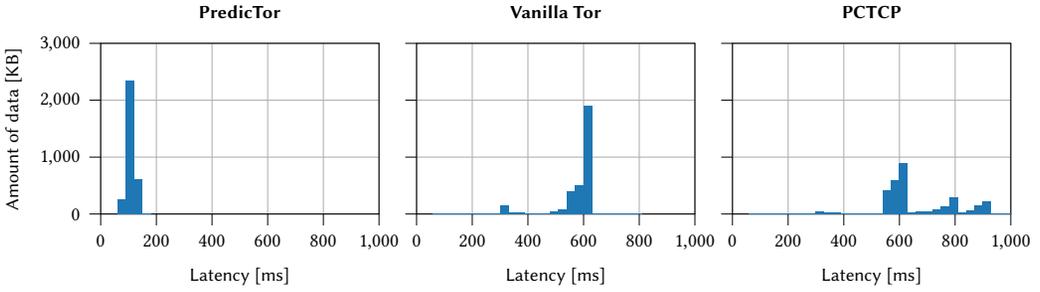

We further compare PredicTor, Tor, and PCTCP
in Figure~\ref{fig:multiple-controllers-backlog} and~\ref{fig:multiple-controllers-latency},
where we display the backlog and latency.
PredicTor succeeds at its primary goal of sustaining a manageable backlog,
especially compared to vanilla Tor and PCTCP.
The importance of this effective congestion control becomes apparent
in Figure~\ref{fig:multiple-controllers-latency},
where we compare histograms for the latencies of received packets.
With an average latency of 93~ms, PredicTor significantly improves on vanilla Tor (553~ms),
and PCTCP (635~ms).
Based on the underlay link latency, the theoretical minimum was at 80~ms.

\begin{table}
    \footnotesize
	\caption{Comparison of latency and throughput in the sample topology.}
	\begin{tabular}{ccccccccc}
		\toprule
		& \multicolumn{3}{c}{\textbf{Mean latency}} && \multicolumn{4}{c}{\textbf{Throughput}} \\
		& \multicolumn{3}{c}{[ms]} && \multicolumn{4}{c}{[KB/s]} \\
        \cmidrule{2-4}\cmidrule{6-9}
		circuit        & PredicTor & Vanilla & PCTCP &\hspace*{.5cm}& PredicTor & Vanilla & PCTCP & \emph{max-min fair} \\
		\midrule
		1              &      91.7 &   534.7 & 601.6 &            & 134.2     &   102.5 & 130.9 & \emph{136.7} \\
		2              &      98.7 &   545.4 & 697.5 &            & 134.1     &   102.5 & 143.9 & \emph{136.7} \\
		3              &      93.7 &   572.1 & 654.2 &            & 134.1     &   197.3 & 130.5 & \emph{136.7} \\
		\midrule
		\textbf{Total} &      93.5 &   553.1 & 635.5 &            & 402.4     &   402.3 & 405.3 & \emph{410.1} \\
		\bottomrule
	\end{tabular}
	\label{tab:multiple-controllers-comparison}
\end{table}

To summarize our findings from this simple network topology,
we present the achieved latency and data rate values per circuit
in Table~\ref{tab:multiple-controllers-comparison}.
Please note that, as mentioned before, the data rates stem from a slightly different setup (scenario~2),
in which circuit~2 does not stop sending.
Otherwise, the data rates would not be easily comparable.
Regarding throughput, the three methods perform very similarly,
with the difference that PredicTor achieves near perfect fairness~($F=0.98$).
Vanilla Tor clearly discriminates circuits~1 and~2 which share a connection~($F=0.68$).
On the other hand, PCTCP, as expected, manages to revise this effect to some extent
in this simple setup~($F=0.95$).

We conclude that PredicTor bears the potential to provide a clear advantage
with respect to latency and fairness.
However, it should be noted that PredicTor also introduces significant complexity
compared to the previous methods.
On the other hand, the optimization problem~\eqref{eq:mpc_full_optim} is convex,
which guarantees a global solution in polynomial time.
For the given scenario, obtaining a solution takes around 10~ms (laptop-grade CPU).
The problem complexity (number of optimization variables and constraints)
grows linearly with the number of circuits per node.
It is therefore expected that also larger topologies can be tackled with this approach in real-time.
However, scaling the network excessively may well render the approach unusable at some point.
This also becomes apparent when considering the overhead induced by the exchange of feedback messages.

\subsection{Impact of Network Complexity} \label{sec:network-complexity}

In the previous subsections, we have demonstrated the general utility of PredicTor
for congestion control in the Tor network.
While this constitutes an important precondition for applying such approaches,
it is not sufficient for thoroughly assessing its potential.
We therefore now go a step further and investigate the extent
to which the observed benefits and drawbacks also apply to more realistic network situations.
In order to do so, we now focus on more complex networks
that are not trivial to comprehend in every simulated detail.
In the course of this evaluation, we put a special focus on the assumptions and trade-offs made in PredicTor.
Note that this evaluation is still explorative in nature.

We first analyze PredicTor's behavior in networks that are significantly larger
and more complex than in Section~\ref{sec:eval-small-network},
but otherwise do not differ much from the simulation assumptions.
In particular, the underlay links still have a uniform, constant latency.

We construct random networks of different size.
In particular, we fix the number of relays at 50,
but vary the number of circuits, ranging from 10 to 1,000.
Each circuit consists of a random sequence of three relays.
We chose this simple network model for several reasons.
First of all, it allows us to easily realize different levels of congestion in the network.
Since we want to compare PredicTor against existing
congestion control mechanisms for the Tor network,
it is desirable to investigate the influence of network load.
And secondly, we do not want to make too strict assumptions on the precise topology.
Instead, we regard a completely random network as a suitable baseline to compare against.
Much more elaborate models of the Tor network do exist~\cite{jansen2012methodically,neverenough-sec2021}.
In fact, our methodology is still influenced by~\cite{neverenough-sec2021},
\eg in that we pay attention to generating a completely new random network
for every single simulation run to avoid statistical bias.

Our focus is mainly on bulk traffic, that is, each circuit transfers and infinite stream of data.
After a lead time, we evaluate the steady state (identified manually by ourselves),
\ie, the last two seconds of simulation time.
For this time span, we again analyze the following metrics:
byte-wise end-to-end latency, throughput, and fairness.
For each data point, we carry out the simulation 25~times (with different random seeds)
and report mean values.
In Section~\ref{sec:eval-traffic-patterns},
we lift the assumption of having only bulk traffic
and consider a mix of bulk and interactive web traffic instead.

\definecolor{color0}{rgb}{0.12156862745098,0.466666666666667,0.705882352941177}
\definecolor{color1}{rgb}{1,0.498039215686275,0.0549019607843137}
\definecolor{color2}{rgb}{0.172549019607843,0.627450980392157,0.172549019607843}

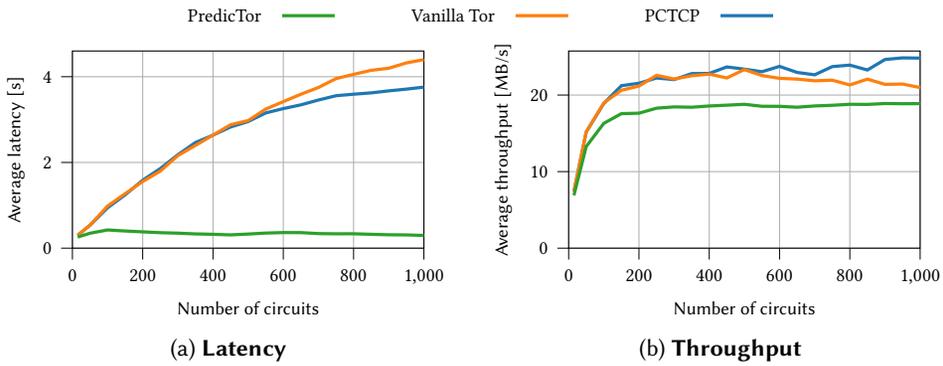
\begin{figure}
\pgfplotsset{width=.45\textwidth,height=.3\textwidth}
\pgfplotsset{
tick label style={font=\scriptsize\sffamily},
label style={font=\scriptsize\sffamily},
legend style={font=\footnotesize},
}
\centering{%
\begin{tikzpicture}[font=\scriptsize]
\node (v) {PredicTor};
\draw[color=color2,very thick] ($(v.east) + (0.15cm,0cm)$) -- +(0.7cm,0cm);

\node[right=1.75cm of v] (s) {Vanilla Tor};
\draw[color=color1,very thick] ($(s.east) + (0.15cm,0cm)$) -- +(0.7cm,0cm);

\node[right=1.75cm of s] (n) {PCTCP};
\draw[color=color0,very thick] ($(n.east) + (0.15cm,0cm)$) -- +(0.7cm,0cm);

\end{tikzpicture} \par \vspace{-1em}
\subfloat[\textbf{Latency}\label{fig:eval-latency}]{%
\begin{tikzpicture}

\definecolor{color0}{rgb}{0.12156862745098,0.466666666666667,0.705882352941177}
\definecolor{color1}{rgb}{1,0.498039215686275,0.0549019607843137}
\definecolor{color2}{rgb}{0.172549019607843,0.627450980392157,0.172549019607843}

\begin{axis}[
tick align=outside,
tick pos=left,
x grid style={white!69.0196078431373!black},
xlabel={Number of circuits},
xmajorgrids,
xmin=0, xmax=1000,
xtick style={color=black},
y grid style={white!69.0196078431373!black},
ylabel={Average latency [s]},
ymajorgrids,
ymin=0, ymax=4.6,
ytick style={color=black}
]
\addplot [very thick, color0]
table {%
15 0.299044722502218
50 0.535693978269878
100 0.941221172151202
150 1.24351065063589
200 1.58578981828109
250 1.85664758655301
300 2.17954898187109
350 2.46389566611187
400 2.63661709528599
450 2.82708599727215
500 2.95366329712657
550 3.15466759127041
600 3.25593299759534
650 3.3434594033175
700 3.45617984804018
750 3.55660813218113
800 3.59199950455213
850 3.62317508593427
900 3.66988179823383
950 3.71063682653802
1000 3.75732001222721
};
\addplot [very thick, color1]
table {%
15 0.299044722502218
50 0.535693978269878
100 0.982220441543117
150 1.27155222003577
200 1.55508911205279
250 1.79629860360977
300 2.15877627163223
350 2.39615352566701
400 2.63777226902002
450 2.87973619614537
500 2.97715839834503
550 3.24514720085837
600 3.41827555061281
650 3.59069395208057
700 3.74845927407568
750 3.95617326098142
800 4.05533961442037
850 4.14905953660745
900 4.19785029787908
950 4.32167517364613
1000 4.39874730407448
};
\addplot [very thick, color2]
table {%
15 0.258259085184475
50 0.347632740683658
100 0.42506508924871
150 0.400645053403525
200 0.380361153350687
250 0.3608423752446
300 0.350063836692996
350 0.332359109330084
400 0.323442112327468
450 0.31001137107408
500 0.329251905567171
550 0.352007928076096
600 0.363998178025215
650 0.362622533277231
700 0.339866661568043
750 0.33514788588009
800 0.3357238903141
850 0.323216417143893
900 0.312335589132088
950 0.309499227651157
1000 0.29658657657354
};
\end{axis}

\end{tikzpicture}}\quad
\subfloat[\textbf{Throughput}\label{fig:eval-throughput}]{%
\begin{tikzpicture}

\definecolor{color0}{rgb}{0.12156862745098,0.466666666666667,0.705882352941177}
\definecolor{color1}{rgb}{1,0.498039215686275,0.0549019607843137}
\definecolor{color2}{rgb}{0.172549019607843,0.627450980392157,0.172549019607843}

\begin{axis}[
tick align=outside,
tick pos=left,
x grid style={white!69.0196078431373!black},
xlabel={Number of circuits},
xmajorgrids,
xmin=0, xmax=1000,
xtick style={color=black},
y grid style={white!69.0196078431373!black},
ylabel={Average throughput [MB/s]},
ymajorgrids,
ymin=0, ymax=25.7542168292999,
ytick style={color=black}
]
\addplot [very thick, color0]
table {%
15 7.41959095001221
50 15.182318687439
100 18.9029188156128
150 21.2229509353638
200 21.5432910919189
250 22.2264776229858
300 22.0205955505371
350 22.8051795959473
400 22.8248891830444
450 23.673113822937
500 23.3841190338135
550 23.0492935180664
600 23.7469654083252
650 22.9545450210571
700 22.6465530395508
750 23.7170448303223
800 23.9105787277222
850 23.2741727828979
900 24.6450576782227
950 24.8566389083862
1000 24.8279056549072
};
\addplot [very thick, color1]
table {%
15 7.41959095001221
50 15.182318687439
100 18.9755830764771
150 20.6231145858765
200 21.175220489502
250 22.5738887786865
300 22.0980091094971
350 22.5225963592529
400 22.7365522384644
450 22.2343139648438
500 23.3292646408081
550 22.5487174987793
600 22.1863460540771
650 22.0854234695435
700 21.8636312484741
750 21.9403324127197
800 21.3212614059448
850 22.0811491012573
900 21.3982000350952
950 21.4402313232422
1000 20.9807367324829
};
\addplot [very thick, color2]
table {%
15 6.9050804901123
50 13.2745921325684
100 16.3022981643677
150 17.5701517868042
200 17.6289765930176
250 18.2852820968628
300 18.4526948547363
350 18.4096472167969
400 18.5739919281006
450 18.6768047332764
500 18.7912533187866
550 18.5337558746338
600 18.5250171661377
650 18.4106255722046
700 18.576594543457
750 18.6619584274292
800 18.79638256073
850 18.7775373458862
900 18.8926698303223
950 18.8654373550415
1000 18.8790678405762
};
\end{axis}

\end{tikzpicture}%
}
\caption{Performance in random networks with 50~relays and variable number of circuits.
    (\ref{fig:eval-latency})~Average byte-wise latency, and
    (\ref{fig:eval-throughput})~average overall throughput (per-run sum of all circuits).} \label{fig:eval-complexity-latency-throughput}
}
\end{figure}

\paragraph{Latency and Throughput}
We first focus on the latency and throughput that is achieved
by each of the three algorithms with growing congestion in the network.
Figure~\ref{fig:eval-complexity-latency-throughput} presents our results.
With respect to latency, we can see that PredicTor offers great potential
to heavily improve on the status quo (see Figure~\ref{fig:eval-latency}).
In particular, by explicitly requiring small queues during optimization,
PredicTor achieves low latency independently of the number of circuits.
In contrast, latency grows indefinitely for denser networks in the case of vanilla Tor and PCTCP.
This is because PredicTor does not \enquote{blindly} send data into the network,
but only if the controller's optimization result allows to do so,
based on local measurements and feedback from adjacent relays.
The resulting lower backlog leads to much lower latencies, even for heavily crowded networks.
In contrast, vanilla Tor and PCTCP have to rely solely
on the state of their local TCP connections
that cannot take into account the state in the network more than one hop down the circuit.
Therefore, they send too much data, leading to significant backlog and latency.
The only way that vanilla Tor and PCTCP could react to congestion on the overall circuit
would be Tor's end-to-end \texttt{SENDME} window mechanism.
However, this window has previously been identified to be too coarse-grained and rigid
to help with efficient congestion control~\cite{DBLP:conf/pet/AlSabahBGGMSV11}.
In contrast to vanilla Tor, PCTCP can deal with extreme congestion slightly better
due to its avoidance of head-of-line blocking in the case of packet loss
that becomes more relevant in these scenarios.

When looking at the achieved throughput, however,
PredicTor cannot fully compete with the traditional approaches.
Over the whole parameter range, it achieves considerably lower overall data rates,
averaging at a disadvantage of around 20\%.
While this is an insight
that was not apparent from the toy scenarios we examined in the previous section,
we attribute it to two root causes:
Firstly, the controller behaves conservatively as far as data rate assignment is concerned.
A circuit is only given a share of the available bandwidth
\emph{after} this was decided to be beneficial in the sense of the MPC optimization problem.
This lack of aggressiveness differentiates PredicTor from the other approaches.
And secondly, just as vanilla Tor and PCTCP,
PredicTor currently uses TCP as the underlying transport protocol
and acts as an additional scheduler on top of it.
In the general case, it will therefore not be able to outperform
approaches that only use TCP without an additional sending limit.
Either way, we can state that the lower average throughput
constitutes a clear trade-off that PredicTor makes in favor of lower latency.

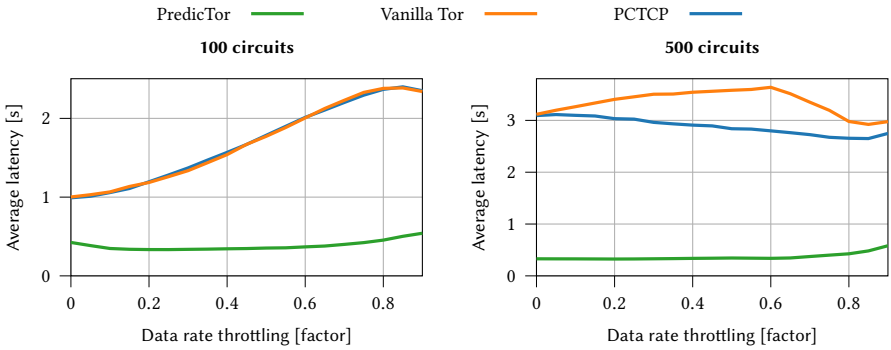
\begin{figure}
    \pgfplotsset{width=.45\textwidth,height=.3\textwidth}
    \pgfplotsset{
    tick label style={font=\scriptsize\sffamily},
    label style={font=\scriptsize\sffamily},
    title style={font=\scriptsize\bfseries\sffamily},
    legend style={font=\footnotesize},
    group/horizontal sep=0.5cm,
    }
    \centering{%
    \begin{tikzpicture}[font=\scriptsize]
    \node (v) {PredicTor};
    \draw[color=color2,very thick] ($(v.east) + (0.15cm,0cm)$) -- +(0.7cm,0cm);

    \node[right=1.75cm of v] (s) {Vanilla Tor};
    \draw[color=color1,very thick] ($(s.east) + (0.15cm,0cm)$) -- +(0.7cm,0cm);

    \node[right=1.75cm of s] (n) {PCTCP};
    \draw[color=color0,very thick] ($(n.east) + (0.15cm,0cm)$) -- +(0.7cm,0cm);
    \end{tikzpicture} \par %
    \begin{tikzpicture}

\definecolor{color0}{rgb}{0.12156862745098,0.466666666666667,0.705882352941177}
\definecolor{color1}{rgb}{1,0.498039215686275,0.0549019607843137}
\definecolor{color2}{rgb}{0.172549019607843,0.627450980392157,0.172549019607843}

\begin{axis}[
tick align=outside,
tick pos=left,
title={100 circuits},
x grid style={white!69.0196078431373!black},
xlabel={Data rate throttling [factor]},
xmajorgrids,
xmin=0, xmax=0.9,
xtick style={color=black},
y grid style={white!69.0196078431373!black},
ylabel={Average latency [s]},
ymajorgrids,
ymin=0, ymax=2.50389771223858,
ytick style={color=black}
]
\addplot [very thick, color0]
table {%
0 0.991302379068386
0.05 1.01105246878164
0.1 1.05770590107814
0.15 1.10967335717214
0.2 1.19508946300661
0.25 1.28140126711834
0.3 1.37065800344118
0.35 1.47080814095071
0.4 1.56820972962224
0.45 1.66972425578652
0.5 1.78506196750569
0.55 1.89900134491253
0.6 2.01296969368338
0.65 2.10570179225535
0.7 2.20198045722565
0.75 2.29511639745386
0.8 2.36826965929847
0.85 2.40055048600099
0.9 2.35139446873554
};
\addplot [very thick, color1]
table {%
0 1.00189893376415
0.05 1.03064572636852
0.1 1.06624297811603
0.15 1.13542103126926
0.2 1.1860384861455
0.25 1.25975511657852
0.3 1.33593980883245
0.35 1.43716418948471
0.4 1.54007320828326
0.45 1.66998911413967
0.5 1.77176727957758
0.55 1.88382156179751
0.6 2.00554386175936
0.65 2.12558734705518
0.7 2.22971738052533
0.75 2.33009172776714
0.8 2.38111897428124
0.85 2.38654168988207
0.9 2.34106552177471
};
\addplot [very thick, color2]
table {%
0 0.424390060587576
0.05 0.38459581141535
0.1 0.347653170503562
0.15 0.337451664352938
0.2 0.333938767293494
0.25 0.333605961249115
0.3 0.336741371816748
0.35 0.339481230871721
0.4 0.343700601676486
0.45 0.346869829087087
0.5 0.353210396977736
0.55 0.356192989908162
0.6 0.36819962929647
0.65 0.378725488439735
0.7 0.399635234809374
0.75 0.422365541063582
0.8 0.454056395325887
0.85 0.502305981426451
0.9 0.54115908870875
};
\end{axis}

\end{tikzpicture} \quad
    \begin{tikzpicture}

\definecolor{color0}{rgb}{0.12156862745098,0.466666666666667,0.705882352941177}
\definecolor{color1}{rgb}{1,0.498039215686275,0.0549019607843137}
\definecolor{color2}{rgb}{0.172549019607843,0.627450980392157,0.172549019607843}

\begin{axis}[
tick align=outside,
tick pos=left,
title={500 circuits},
x grid style={white!69.0196078431373!black},
xlabel={Data rate throttling [factor]},
xmajorgrids,
xmin=0, xmax=0.9,
xtick style={color=black},
y grid style={white!69.0196078431373!black},
ylabel={Average latency [s]},
ymajorgrids,
ymin=0, ymax=3.8039669048205,
ytick style={color=black}
]
\addplot [very thick, color0]
table {%
0 3.09300544657357
0.05 3.11229337114211
0.1 3.09724436453987
0.15 3.08475707397585
0.2 3.03322575373842
0.25 3.02415905579444
0.3 2.96169053235711
0.35 2.93325745614155
0.4 2.90866770550707
0.45 2.89367826361721
0.5 2.83899067352719
0.55 2.83216888844395
0.6 2.79791028919564
0.65 2.7627932718751
0.7 2.72504184533715
0.75 2.67357846810085
0.8 2.65490003692973
0.85 2.64910469781964
0.9 2.74831458862766
};
\addplot [very thick, color1]
table {%
0 3.11663521554462
0.05 3.1958813048803
0.1 3.26451767003247
0.15 3.33540580677194
0.2 3.40578960846582
0.25 3.45503620648927
0.3 3.50466326223805
0.35 3.50821929827107
0.4 3.5431088444158
0.45 3.56166282718828
0.5 3.58094808934818
0.55 3.59663918323346
0.6 3.63828572964942
0.65 3.51453028737357
0.7 3.35103051944326
0.75 3.19375314563264
0.8 2.97906384938851
0.85 2.92232534261645
0.9 2.97599669721174
};
\addplot [very thick, color2]
table {%
0 0.329202326372134
0.05 0.328134558387465
0.1 0.327693795575727
0.15 0.326309203726063
0.2 0.324662226227857
0.25 0.32621526560061
0.3 0.329511298427422
0.35 0.33299727914894
0.4 0.337430260145817
0.45 0.340541286324262
0.5 0.344518864995326
0.55 0.341350744370663
0.6 0.338149120046557
0.65 0.345394145559009
0.7 0.373053308750044
0.75 0.399923786371862
0.8 0.425679037060434
0.85 0.481400565771782
0.9 0.581246978232479
};
\end{axis}

\end{tikzpicture}
    \caption{Impact of application-layer throttling on latency in networks with 50 relays.
             Note that reducing the data rate does not enable vanilla Tor and PCTCP
             to achieve latency values as low as PredicTor.}
    \label{fig:eval-complexity-throttled}
}
\end{figure}

Putting this relationship into perspective, one might argue
that the lower latency is not an achievement of PredicTor itself,
but simply an artifact and a consequence of the lower throughput,
because the lower data rates result in smaller queues.
However, this is not the case as another experiment shows:
Out of the previously explored parameter space,
we chose two scenarios that are representative
for a low and high degree of congestion in the network,
respectively (100 and 500 circuits, with 50 relays).
For both of these scenarios, we introduced an artificial,
\emph{application-layer} throttling mechanism,
reducing the amount of available bandwidth each relay can use.
We varied this throttling factor between 0.0 (no throttling at all)
and 0.9 (only 10\% of bandwidth remain) and recorded the achieved latency.
The assumption was that artificially lowering the data rates of vanilla Tor and PCTCP
would already be enough to achieve also lower latency even for these mechanisms.
Figure~\ref{fig:eval-complexity-throttled} reveals that the opposite is true,
for both the heavily and less congested networks.
In fact, lowering the data rates mostly even leads to
an \emph{increase} in latency with vanilla Tor and PCTCP.
To understand this behavior, we have to emphasize that each data transfer through the Tor network
does not consist of only one single TCP connection, but denotes a multi-hop data transfer.
Lowering the data rates therefore does not automatically
lead to a substantial reduction of backlog.
Instead, the packets are queued at the application layer
and experience the same throttling when being forwarded.
We can thus conclude that the low latency is in fact an achievement of PredicTor
and not only a side effect of the lower data rates.

\begin{figure}
\pgfplotsset{width=.45\textwidth,height=.3\textwidth}
\pgfplotsset{
tick label style={font=\scriptsize\sffamily},
label style={font=\scriptsize\sffamily},
legend style={font=\footnotesize},
}
\centering{%
\begin{tikzpicture}[font=\scriptsize]
\node (v) {PredicTor};
\draw[color=color2,very thick] ($(v.east) + (0.15cm,0cm)$) -- +(0.7cm,0cm);

\node[right=1.75cm of v] (s) {Vanilla Tor};
\draw[color=color1,very thick] ($(s.east) + (0.15cm,0cm)$) -- +(0.7cm,0cm);

\node[right=1.75cm of s] (n) {PCTCP};
\draw[color=color0,very thick] ($(n.east) + (0.15cm,0cm)$) -- +(0.7cm,0cm);

\node[right=1.75cm of n] (f) {max-min fairness};
\draw[color=black,very thick] ($(f.east) + (0.15cm,0cm)$) -- +(0.7cm,0cm);
\end{tikzpicture} \par \vspace{-1em}
\subfloat[\textbf{Data rate distribution} (750 circuits)\label{fig:eval-fairness-cdf}]{%
\input{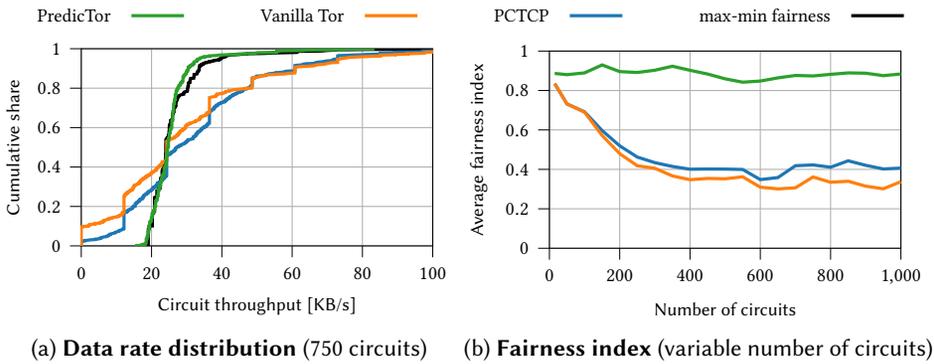}}
\subfloat[\textbf{Fairness index} (variable number of circuits)\label{fig:eval-fairness-index}]{%
\begin{tikzpicture}

\definecolor{color0}{rgb}{0.12156862745098,0.466666666666667,0.705882352941177}
\definecolor{color1}{rgb}{1,0.498039215686275,0.0549019607843137}
\definecolor{color2}{rgb}{0.172549019607843,0.627450980392157,0.172549019607843}

\begin{axis}[
tick align=outside,
tick pos=left,
x grid style={white!69.0196078431373!black},
xlabel={Number of circuits},
xmajorgrids,
xmin=0, xmax=1000,
xtick style={color=black},
y grid style={white!69.0196078431373!black},
ylabel={Average fairness index},
ymajorgrids,
ymin=0, ymax=1,
ytick style={color=black}
]
\addplot [very thick, color0]
table {%
15 0.83490145882353
50 0.731778590718854
100 0.691761347837625
150 0.597023937049325
200 0.519476363135345
250 0.462362013232698
300 0.434069854907524
350 0.414737717112098
400 0.4008008349661
450 0.401801632473721
500 0.401427849358274
550 0.399491289483645
600 0.348130403611184
650 0.358257650697643
700 0.418521216644467
750 0.422443008669109
800 0.410868355767722
850 0.443432731307671
900 0.420554685008957
950 0.402098429358723
1000 0.406426618839411
};
\addplot [very thick, color1]
table {%
15 0.83490145882353
50 0.731778590718854
100 0.689801638213358
150 0.571790118396175
200 0.47973885187055
250 0.418318323618682
300 0.405593226184272
350 0.366627585701676
400 0.347628673732801
450 0.353882890238343
500 0.352221361665069
550 0.362203726303738
600 0.309574672342731
650 0.300933734264499
700 0.306580451782325
750 0.361282346446575
800 0.334654686154452
850 0.339730994421345
900 0.314725912800098
950 0.301478522183912
1000 0.337576772001455
};
\addplot [very thick, color2]
table {%
15 0.886644685803921
50 0.880333865886728
100 0.889410174655477
150 0.929668102844158
200 0.895969096698278
250 0.891416884229983
300 0.902884286805391
350 0.923069725055597
400 0.902551561885053
450 0.88361677148268
500 0.859259447947163
550 0.842186709148405
600 0.847965188262477
650 0.864067297025777
700 0.876335739632024
750 0.873837540008447
800 0.881937998506905
850 0.889360124313356
900 0.887777704134275
950 0.875390480477718
1000 0.883844431044974
};
\end{axis}

\end{tikzpicture}}\quad
\caption{Throughput fairness in random networks with 50~relays.
    (\ref{fig:eval-fairness-cdf})~Data rate CDF plot of an example run (750~circuits), and
    (\ref{fig:eval-fairness-index})~fairness index over varying number of circuits.} \label{fig:eval-complexity-fairness}
}
\end{figure}

\paragraph{Fairness}
Another central promise of PredicTor is the achievement of much better fairness,
based on the notion of max-min fairness.
We now evaluate the degree to which PredicTor can realize fairness also in complex networks.
Our results are based on the same simulation runs as for the latency and throughput evaluation.

In Figure~\ref{fig:eval-fairness-cdf},
we show an individual simulation run with 750~concurrent circuits
as a CDF plot of the data rates to visually inspect the fairness.
We also included the max-min fair rate distribution as a baseline.
As can be seen, vanilla Tor and PCTCP give most circuits either too low or too high data rates.
In contrast, PredicTor very closely approximates max-min fairness.
The only deviation that can visually be identified
is that several circuits use less bandwidth than optimal max-min would allow them to.
This is in line with our previous observation
that the traffic generated by PredicTor is rather conservative
and the overall data rate tends to be lower than with the traditional approaches.

We validated and generalized the insight gained from the single simulation runs
by calculating the fairness index~$F$ for varying circuit numbers.
The plot in Figure~\ref{fig:eval-fairness-index} reveals that PredicTor is highly effective at ensuring fairness,
even in situations in which the network is heavily congested.
The explicit max-min fairness formulation in PredicTor's optimization goal
consistently causes around 90\% of the network traffic to adhere to max-min fairness.
In contrast, vanilla Tor and PCTCP generally generate much less fair traffic.
This becomes especially apparent the more congested the network is.
Again, PCTCP performs slightly better than vanilla Tor,
but still cannot clearly surpass the threshold of around $F=0.4$
if there is considerable congestion in the network.
We can also see that, if there is only a little congestion,
both of the traditional approaches generate fairer traffic.
This, however, is not because they can \emph{ensure} this behavior in any kind.
Instead, the lack of congestion also implies that
more circuits can fully utilize the available bandwidth on their path.
Therefore, a larger share of circuits can be regarded as behaving in a fair way.

\subsection{Impact of Model Assumptions} \label{sec:eval-model-assumptions}

As shown in the previous subsection, PredicTor is able to improve on both, latency and fairness,
even in complex, \enquote{crowded} networks.
However, even if the concurrent data transmissions of many circuits in these simulations
added a considerable degree of randomness to the network behavior,
the scenario is still specific to the assumptions made in PredicTor's system model.
Most importantly, the underlay link latencies exactly match the values
that are used for calculation in the model.
In reality, this would not be the case
as many influences outside our model affect the connection.
Such factors might include cross traffic on the Internet, routing topology changes and others.

The system model describing the expected network behavior is crucial
to the functioning of model predictive control approaches like PredicTor.
It is therefore important to investigate to which degree the overall system
is susceptible to deviations from the model.
In this section, we thus focus on the robustness of PredicTor
in the face of network behavior differing from PredicTor's expectations drawn from its system model.
We note that the strongest assumption made by PredicTor is
that the latency of the underlay links can reliably be known in advance.
Therefore, we now evaluate PredicTor's behavior if this assumption is violated.

For this, we employ a similar setup as in Section~\ref{sec:network-complexity}.
However, we now do not simulate uniform link latency.
In contrast, we introduce a fuzziness factor $f$.
The fuzziness $f$ defines the uncertainty in link latency as follows:
If PredicTor's system model expects a link latency of $l$,
the link latency is instead chosen uniformly at random from the interval
$[\max(l \cdot (1-f), \epsilon ) ;  l \cdot (1+f)]$.
As a consequence of this construction, the larger the fuzziness value $f$,
the larger will be the average deviation of the underlay link latencies from PredicTor's system model.
This gives us a concise parameter to evaluate the robustness of PredicTor against a system model mismatch.
As a technical detail, we introduce a lower bound
of some arbitrarily small value $\epsilon > 0$ for the latencies
to avoid negative and zero-valued latencies.
For fuzziness values~$f > 1$, this by design shifts the distribution towards higher latencies.

We now fix the number of relays and circuits
to evaluate PredicTor's robustness by varying the link fuzziness.
Since the analysis without link latency deviation has revealed
that the performance differs depending on how crowded the network is,
we carry out the following evaluation twice:
Firstly, with a circuit number of~100, representing a network situation with little congestion,
and secondly, with 500~circuits, which induces much more congestion in the network.
Again, each trial is repeated 25~times with newly generated network topologies.

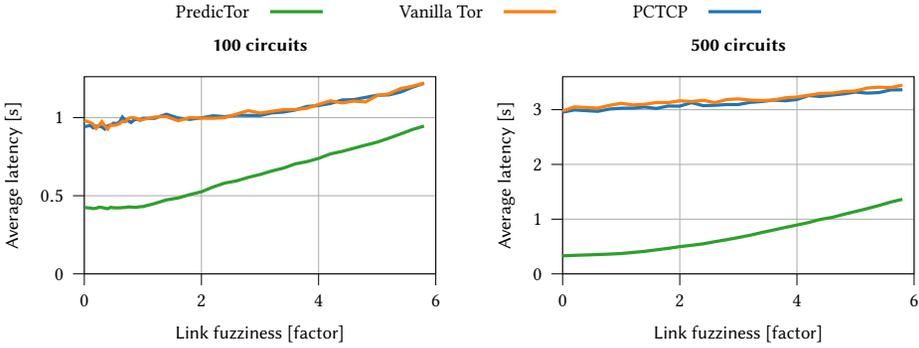
\begin{figure}
    \pgfplotsset{width=.45\textwidth,height=.3\textwidth}
    \pgfplotsset{
    tick label style={font=\scriptsize\sffamily},
    label style={font=\scriptsize\sffamily},
    title style={font=\scriptsize\bfseries\sffamily},
    legend style={font=\footnotesize},
    group/horizontal sep=0.5cm,
    }
    \centering{%
    \begin{tikzpicture}[font=\scriptsize]
    \node (v) {PredicTor};
    \draw[color=color2,very thick] ($(v.east) + (0.15cm,0cm)$) -- +(0.7cm,0cm);

    \node[right=1.75cm of v] (s) {Vanilla Tor};
    \draw[color=color1,very thick] ($(s.east) + (0.15cm,0cm)$) -- +(0.7cm,0cm);

    \node[right=1.75cm of s] (n) {PCTCP};
    \draw[color=color0,very thick] ($(n.east) + (0.15cm,0cm)$) -- +(0.7cm,0cm);
    \end{tikzpicture} \par %
    \begin{tikzpicture}

\definecolor{color0}{rgb}{0.12156862745098,0.466666666666667,0.705882352941177}
\definecolor{color1}{rgb}{1,0.498039215686275,0.0549019607843137}
\definecolor{color2}{rgb}{0.172549019607843,0.627450980392157,0.172549019607843}

\begin{axis}[
tick align=outside,
tick pos=left,
title={100 circuits},
x grid style={white!69.0196078431373!black},
xlabel={Link fuzziness [factor]},
xmajorgrids,
xmin=0, xmax=6,
xtick style={color=black},
y grid style={white!69.0196078431373!black},
ylabel={Average latency [s]},
ymajorgrids,
ymin=0, ymax=1.26169345070702,
ytick style={color=black}
]
\addplot [very thick, color0]
table {%
0 0.941221172151202
0.1 0.949948899165171
0.15 0.93502794137423
0.2 0.945804348688445
0.25 0.943550887825745
0.3 0.94092886361276
0.35 0.926360264723121
0.4 0.946931038732715
0.45 0.952529475818616
0.5 0.964347469305575
0.55 0.960980715706643
0.6 0.970757441455538
0.65 1.0034589734771
0.7 0.982619126464723
0.75 0.97894201194261
0.8 0.968741971525059
0.85 0.984611035943479
0.9 0.989911726959788
0.95 0.983532965196267
1 0.995694560058797
1.2 0.995140972139793
1.4 1.0216878214149
1.6 0.997449426322778
1.8 0.988563601212504
2 0.99859717194866
2.2 1.01171301424486
2.4 1.005813848684
2.6 1.01503054113501
2.8 1.01385835787565
3 1.01368487869378
3.2 1.03007604974783
3.4 1.03537359034943
3.6 1.04829824801613
3.8 1.07109320338145
4 1.07731250633384
4.2 1.09017648761391
4.4 1.11269725927347
4.6 1.11505811677474
4.8 1.12919389728146
5 1.14318624181379
5.2 1.14535802417822
5.4 1.16328196129744
5.6 1.19435595749255
5.8 1.21930908046278
};
\addplot [very thick, color1]
table {%
0 0.982220441543117
0.1 0.96885834867257
0.15 0.958733967965397
0.2 0.931115923981456
0.25 0.949836253825008
0.3 0.973653437305091
0.35 0.948356013016278
0.4 0.926839800841138
0.45 0.953573310632553
0.5 0.950047229248035
0.55 0.953357413919855
0.6 0.960486218470172
0.65 0.977312097220551
0.7 0.97291337209942
0.75 0.993967548277341
0.8 1.0007027608346
0.85 0.998140770349952
0.9 0.993789776838654
0.95 0.981670837444948
1 0.987588221556698
1.2 1.0037082754805
1.4 1.00684568856503
1.6 0.979585993295916
1.8 1.00015316294616
2 0.996000541355239
2.2 0.995649977776811
2.4 0.999702821510038
2.6 1.02257014975284
2.8 1.04468247686917
3 1.02908830727746
3.2 1.04020309512784
3.4 1.05067065434003
3.6 1.05142369430995
3.8 1.05985937312881
4 1.08434671646171
4.2 1.10736191875163
4.4 1.09397542314368
4.6 1.10554254798842
4.8 1.10015352518972
5 1.14108898863521
5.2 1.15207151969158
5.4 1.18593320969864
5.6 1.20022128681026
5.8 1.22146449895778
};
\addplot [very thick, color2]
table {%
0 0.424683630058685
0.1 0.421187032928028
0.15 0.417937987145192
0.2 0.419427942003404
0.25 0.426587803799339
0.3 0.425165798941256
0.35 0.420180570138234
0.4 0.416885463972981
0.45 0.425834063614161
0.5 0.422917343277346
0.55 0.422354666408648
0.6 0.422485982187684
0.65 0.4242416958592
0.7 0.425231861944133
0.75 0.428165034982148
0.8 0.427681258853599
0.85 0.425359897400683
0.9 0.426469789562327
0.95 0.429808943467358
1 0.430516039425354
1.2 0.44992386616332
1.4 0.47309913888652
1.6 0.485514891349939
1.8 0.507105712595994
2 0.525414487777711
2.2 0.555993353903228
2.4 0.581336525679101
2.6 0.595142920736617
2.8 0.61743822759783
3 0.635712318598434
3.2 0.65821770648635
3.4 0.676827755874431
3.6 0.703707315108338
3.8 0.718136510791896
4 0.738936257729378
4.2 0.767513806736558
4.4 0.783580523658113
4.6 0.803759076780297
4.8 0.823748277641991
5 0.842923491027132
5.2 0.86821278165199
5.4 0.895248555549828
5.6 0.923535364581773
5.8 0.945741711495104
};
\end{axis}

\end{tikzpicture} \quad
    \begin{tikzpicture}

\definecolor{color0}{rgb}{0.12156862745098,0.466666666666667,0.705882352941177}
\definecolor{color1}{rgb}{1,0.498039215686275,0.0549019607843137}
\definecolor{color2}{rgb}{0.172549019607843,0.627450980392157,0.172549019607843}

\begin{axis}[
tick align=outside,
tick pos=left,
title={500 circuits},
x grid style={white!69.0196078431373!black},
xlabel={Link fuzziness [factor]},
xmajorgrids,
xmin=0, xmax=6,
xtick style={color=black},
y grid style={white!69.0196078431373!black},
ylabel={Average latency [s]},
ymajorgrids,
ymin=0, ymax=3.60264757958682,
ytick style={color=black}
]
\addplot [very thick, color0]
table {%
0 2.95366329712657
0.2 2.99543778475366
0.6 2.96939349627964
0.8 3.01196385913161
1 3.02706775647383
1.2 3.02898663647053
1.4 3.04893001923824
1.6 3.02208383692383
1.8 3.06700811263096
2 3.06449497091536
2.2 3.12895590955488
2.4 3.07345836729308
2.6 3.0802136748667
2.8 3.09364905901315
3 3.09384812487945
3.2 3.13347470284029
3.4 3.14945842379076
3.6 3.17523631808805
3.8 3.16257292061273
4 3.18578332680062
4.2 3.2580826778629
4.4 3.24055103971709
4.6 3.26670444191012
4.8 3.28934235693927
5 3.3222155442156
5.2 3.30283943598297
5.4 3.31251132211185
5.6 3.36089854576576
5.8 3.36344173400168
};
\addplot [very thick, color1]
table {%
0 2.97715839834503
0.2 3.05076503012835
0.6 3.02967437687846
0.8 3.07900171321467
1 3.1168383869426
1.2 3.08683949359016
1.4 3.09792590998938
1.6 3.13101856729774
1.8 3.12933695003689
2 3.16254576080781
2.2 3.14882875193411
2.4 3.17158159451499
2.6 3.1277642139505
2.8 3.18510625027575
3 3.19695030265863
3.2 3.17596421791075
3.4 3.16661998092645
3.6 3.17629221818482
3.8 3.21698423055081
4 3.23177429756764
4.2 3.26479861475393
4.4 3.29766206975441
4.6 3.30142383720407
4.8 3.3324941301399
5 3.34250488240022
5.2 3.3956857381233
5.4 3.41090675785165
5.6 3.40484190911056
5.8 3.44677051908223
};
\addplot [very thick, color2]
table {%
0 0.329229308990534
0.2 0.338388810818506
0.6 0.35275336342852
0.8 0.36134267950478
1 0.371840438164674
1.2 0.391352311722407
1.4 0.411651227219439
1.6 0.439325595399424
1.8 0.465656882066578
2 0.498002209652821
2.2 0.522147252338517
2.4 0.549117319660439
2.6 0.588575316201962
2.8 0.624496127406018
3 0.663258415919269
3.2 0.705947173037105
3.4 0.755660050301644
3.6 0.801278322284116
3.8 0.847995563181889
4 0.892640995632441
4.2 0.938579520478866
4.4 0.995195467369284
4.6 1.03210842251424
4.8 1.08795027701959
5 1.13944000069892
5.2 1.19238142755918
5.4 1.24937689601181
5.6 1.31178040135225
5.8 1.3619846986989
};
\end{axis}

\end{tikzpicture}
    \caption{Impact of link latency deviation from the system model on achieved latency, in random networks with 50~relays.}
    \label{fig:eval-fuzziness-latency}
}
\end{figure}
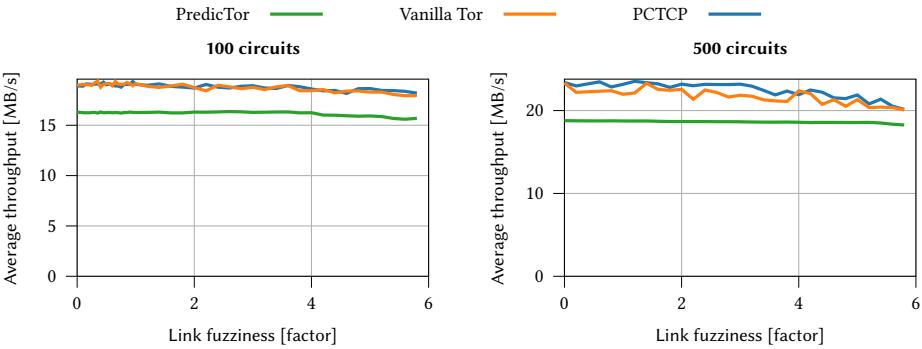
\begin{figure}
    \pgfplotsset{width=.45\textwidth,height=.3\textwidth}
    \pgfplotsset{
    tick label style={font=\scriptsize\sffamily},
    label style={font=\scriptsize\sffamily},
    title style={font=\scriptsize\bfseries\sffamily},
    legend style={font=\footnotesize},
    group/horizontal sep=0.5cm,
    }
    \centering{%
    \begin{tikzpicture}[font=\scriptsize]
    \node (v) {PredicTor};
    \draw[color=color2,very thick] ($(v.east) + (0.15cm,0cm)$) -- +(0.7cm,0cm);

    \node[right=1.75cm of v] (s) {Vanilla Tor};
    \draw[color=color1,very thick] ($(s.east) + (0.15cm,0cm)$) -- +(0.7cm,0cm);

    \node[right=1.75cm of s] (n) {PCTCP};
    \draw[color=color0,very thick] ($(n.east) + (0.15cm,0cm)$) -- +(0.7cm,0cm);
    \end{tikzpicture} \par %
    \begin{tikzpicture}

\definecolor{color0}{rgb}{0.12156862745098,0.466666666666667,0.705882352941177}
\definecolor{color1}{rgb}{1,0.498039215686275,0.0549019607843137}
\definecolor{color2}{rgb}{0.172549019607843,0.627450980392157,0.172549019607843}

\begin{axis}[
tick align=outside,
tick pos=left,
title={100 circuits},
x grid style={white!69.0196078431373!black},
xlabel={Link fuzziness [factor]},
xmajorgrids,
xmin=0, xmax=6,
xtick style={color=black},
y grid style={white!69.0196078431373!black},
ylabel={Average throughput [MB/s]},
ymajorgrids,
ymin=0, ymax=19.5788279972076,
ytick style={color=black}
]
\addplot [very thick, color0]
table {%
0 18.9029188156128
0.1 18.8817844390869
0.15 19.1040515899658
0.2 19.0914659500122
0.25 19.0591707229614
0.3 19.0734186172485
0.35 19.1073760986328
0.4 19.0871915817261
0.45 19.2866621017456
0.5 19.0173768997192
0.55 19.1285104751587
0.6 19.0102529525757
0.65 18.8955574035645
0.7 18.8943700790405
0.75 18.7832365036011
0.8 19.0411233901978
0.85 19.0565586090088
0.9 19.0829172134399
0.95 19.3424663543701
1 19.0567960739136
1.2 18.9684591293335
1.4 19.105001449585
1.6 18.8620748519897
1.8 18.7865610122681
2 18.6934747695923
2.2 19.0489597320557
2.4 18.7799119949341
2.6 18.6958494186401
2.8 18.8699111938477
3 18.9235782623291
3.2 18.7022609710693
3.4 18.6504936218262
3.6 18.953498840332
3.8 18.8226556777954
4 18.6051378250122
4.2 18.4187278747559
4.4 18.4548225402832
4.6 18.164402961731
4.8 18.6314964294434
5 18.6464567184448
5.2 18.4557723999023
5.4 18.4412870407104
5.6 18.3662481307983
5.8 18.1895742416382
};
\addplot [very thick, color1]
table {%
0 18.9755830764771
0.1 19.0577459335327
0.15 19.0197515487671
0.2 19.0173768997192
0.25 18.9584856033325
0.3 19.2085361480713
0.35 19.3892469406128
0.4 18.768988609314
0.45 19.0537090301514
0.5 19.1375341415405
0.55 19.152494430542
0.6 18.9100427627563
0.65 19.3484029769897
0.7 19.1002521514893
0.75 19.0712814331055
0.8 19.0738935470581
0.85 19.2470054626465
0.9 19.1413335800171
0.95 18.9292774200439
1 19.1249485015869
1.2 18.8687238693237
1.4 18.7575902938843
1.6 18.9036312103271
1.8 19.0736560821533
2 18.7440547943115
2.2 18.4139785766602
2.4 18.9751081466675
2.6 18.8174314498901
2.8 18.6279344558716
3 18.76780128479
3.2 18.5362730026245
3.4 18.8145818710327
3.6 18.9349765777588
3.8 18.4294137954712
4 18.4294137954712
4.2 18.5486211776733
4.4 18.2159328460693
4.6 18.3947439193726
4.8 18.4059047698975
5 18.2893095016479
5.2 18.2907342910767
5.4 18.0544567108154
5.6 17.9340620040894
5.8 17.9535341262817
};
\addplot [very thick, color2]
table {%
0 16.3045683288574
0.1 16.2454395675659
0.15 16.2405857849121
0.2 16.2365393829346
0.25 16.260380859375
0.3 16.28246509552
0.35 16.2046335983276
0.4 16.320032043457
0.45 16.236064453125
0.5 16.2543207550049
0.55 16.2630499649048
0.6 16.2415546417236
0.65 16.2572463226318
0.7 16.2492010116577
0.75 16.2060583877563
0.8 16.2626035308838
0.85 16.2615681838989
0.9 16.3124331665039
0.95 16.2713707351685
1 16.2636103820801
1.2 16.2741158294678
1.4 16.3042358779907
1.6 16.2253025436401
1.8 16.2268983078003
2 16.3062400817871
2.2 16.2962570571899
2.4 16.3224351882935
2.6 16.3576084899902
2.8 16.3342324447632
3 16.2799954605103
3.2 16.2956776428223
3.4 16.3176383972168
3.6 16.3223781967163
3.8 16.2310111999512
4 16.2450406265259
4.2 16.0109857177734
4.4 15.9966333389282
4.6 15.9418264389038
4.8 15.8999851226807
5 15.9237126159668
5.2 15.8639664459229
5.4 15.6723987579346
5.6 15.5976258087158
5.8 15.690731048584
};
\end{axis}

\end{tikzpicture} \quad
    \begin{tikzpicture}

\definecolor{color0}{rgb}{0.12156862745098,0.466666666666667,0.705882352941177}
\definecolor{color1}{rgb}{1,0.498039215686275,0.0549019607843137}
\definecolor{color2}{rgb}{0.172549019607843,0.627450980392157,0.172549019607843}

\begin{axis}[
tick align=outside,
tick pos=left,
title={500 circuits},
x grid style={white!69.0196078431373!black},
xlabel={Link fuzziness [factor]},
xmajorgrids,
xmin=0, xmax=6,
xtick style={color=black},
y grid style={white!69.0196078431373!black},
ylabel={Average throughput [MB/s]},
ymajorgrids,
ymin=0, ymax=23.7925596199036,
ytick style={color=black}
]
\addplot [very thick, color0]
table {%
0 23.3841190338135
0.2 22.9733047485352
0.6 23.4655694961548
0.8 22.8505353927612
1 23.1689758300781
1.2 23.529447555542
1.4 23.3627471923828
1.6 23.2197933197021
1.8 22.7973432540894
2 23.1889228820801
2.2 22.997763633728
2.4 23.164701461792
2.6 23.1416673660278
2.8 23.1376304626465
3 23.1848859786987
3.2 22.946946144104
3.4 22.4335470199585
3.6 21.8902273178101
3.8 22.3516216278076
4 21.9156360626221
4.2 22.480565071106
4.4 22.2162666320801
4.6 21.532130241394
4.8 21.44260597229
5 21.8807287216187
5.2 20.8002634048462
5.4 21.3723163604736
5.6 20.5233793258667
5.8 20.147234916687
};
\addplot [very thick, color1]
table {%
0 23.3292646408081
0.2 22.1972694396973
0.6 22.3385610580444
0.8 22.3908033370972
1 21.9671659469604
1.2 22.1227054595947
1.4 23.2919826507568
1.6 22.5482425689697
1.8 22.4468450546265
2 22.5610656738281
2.2 21.357593536377
2.4 22.482702255249
2.6 22.1851587295532
2.8 21.6335277557373
3 21.8486709594727
3.2 21.7299385070801
3.4 21.2949028015137
3.6 21.1635847091675
3.8 21.1030311584473
4 22.3841543197632
4.2 22.0441045761108
4.4 20.7454090118408
4.6 21.31556224823
4.8 20.5134057998657
5 21.3162746429443
5.2 20.3611907958984
5.4 20.4141454696655
5.6 20.3657026290894
5.8 20.0805072784424
};
\addplot [very thick, color2]
table {%
0 18.7842623519897
0.2 18.7648092269897
0.6 18.7484526443481
0.8 18.7602214050293
1 18.7456505584717
1.2 18.7373012924194
1.4 18.7471133422852
1.6 18.7064878463745
1.8 18.6786474609375
2 18.6799202728271
2.2 18.6718179702759
2.4 18.6718464660645
2.6 18.6575415802002
2.8 18.6563447570801
3 18.6458013153076
3.2 18.6187683105469
3.4 18.6065531158447
3.6 18.6019747924805
3.8 18.6132781219482
4 18.5886102676392
4.2 18.5491625976562
4.4 18.5716267776489
4.6 18.564483833313
4.8 18.5546337890625
5 18.5570939254761
5.2 18.5750747680664
5.4 18.4997983932495
5.6 18.352237701416
5.8 18.2672062683105
};
\end{axis}

\end{tikzpicture}
    \caption{Impact of link latency deviation from the system model on achieved throughput, in random networks with 50~relays. The throughput values are obtained as the sum of all circuits.}
    \label{fig:eval-fuzziness-throughput}
}
\end{figure}
\begin{figure}
    \pgfplotsset{width=.45\textwidth,height=.3\textwidth}
    \pgfplotsset{
    tick label style={font=\scriptsize\sffamily},
    label style={font=\scriptsize\sffamily},
    title style={font=\scriptsize\bfseries\sffamily},
    legend style={font=\footnotesize},
    group/horizontal sep=0.5cm,
    }
    \centering{%
    \begin{tikzpicture}[font=\scriptsize]
    \node (v) {PredicTor};
    \draw[color=color2,very thick] ($(v.east) + (0.15cm,0cm)$) -- +(0.7cm,0cm);

    \node[right=1.75cm of v] (s) {Vanilla Tor};
    \draw[color=color1,very thick] ($(s.east) + (0.15cm,0cm)$) -- +(0.7cm,0cm);

    \node[right=1.75cm of s] (n) {PCTCP};
    \draw[color=color0,very thick] ($(n.east) + (0.15cm,0cm)$) -- +(0.7cm,0cm);
    \end{tikzpicture} \par %
    \begin{tikzpicture}

\definecolor{color0}{rgb}{0.12156862745098,0.466666666666667,0.705882352941177}
\definecolor{color1}{rgb}{1,0.498039215686275,0.0549019607843137}
\definecolor{color2}{rgb}{0.172549019607843,0.627450980392157,0.172549019607843}

\begin{axis}[
tick align=outside,
tick pos=left,
title={100 circuits},
x grid style={white!69.0196078431373!black},
xlabel={Link fuzziness [factor]},
xmajorgrids,
xmin=0, xmax=6,
xtick style={color=black},
y grid style={white!69.0196078431373!black},
ylabel={Average fairness index},
ymajorgrids,
ymin=0, ymax=1,
ytick style={color=black}
]
\addplot [very thick, color0]
table {%
0 0.691761347837625
0.1 0.699718638507778
0.15 0.675442870393851
0.2 0.671932142575464
0.25 0.677680803475529
0.3 0.699554787213648
0.35 0.683465143960583
0.4 0.693792272362767
0.45 0.693644817180613
0.5 0.703614773983748
0.55 0.699516467556637
0.6 0.708266778809278
0.65 0.672931552202677
0.7 0.669203103733029
0.75 0.698969384275459
0.8 0.710144619689934
0.85 0.685829433689713
0.9 0.676214761656677
0.95 0.667558376391728
1 0.67526741472172
1.2 0.675249144710786
1.4 0.685850124365878
1.6 0.687944099205078
1.8 0.718936960618119
2 0.697546406924352
2.2 0.688131555691387
2.4 0.698096523099379
2.6 0.704596708784262
2.8 0.699197568033614
3 0.680038708311595
3.2 0.69956027766756
3.4 0.667565155978994
3.6 0.704339262246256
3.8 0.703080394774425
4 0.693832123840102
4.2 0.690385756944492
4.4 0.692265278687436
4.6 0.712201467911313
4.8 0.708462931870957
5 0.74374898196004
5.2 0.715391205790301
5.4 0.707119498051085
5.6 0.724874072474384
5.8 0.730666223098881
};
\addplot [very thick, color1]
table {%
0 0.689801638213358
0.1 0.66819943262616
0.15 0.672444184386012
0.2 0.682540848186497
0.25 0.679161360149838
0.3 0.667795608556506
0.35 0.683129568871943
0.4 0.704959095052665
0.45 0.663724597941295
0.5 0.663247446131974
0.55 0.667685502914857
0.6 0.683082692270359
0.65 0.67224332388161
0.7 0.691546109713545
0.75 0.66542512653829
0.8 0.666085566309634
0.85 0.663950041376925
0.9 0.63389520372371
0.95 0.683575317914273
1 0.659329942172949
1.2 0.688373855631015
1.4 0.682194191629951
1.6 0.681840934226692
1.8 0.670161943684798
2 0.68221665385915
2.2 0.679351589761131
2.4 0.676024033972473
2.6 0.685536013952733
2.8 0.703198635994029
3 0.708241749689787
3.2 0.674135639380851
3.4 0.69071389196666
3.6 0.697174907682506
3.8 0.699175933415612
4 0.742497737835897
4.2 0.686023646895208
4.4 0.727622391405511
4.6 0.727626615006576
4.8 0.712828093870384
5 0.721375135916071
5.2 0.727732170823417
5.4 0.725193114258858
5.6 0.727958456635791
5.8 0.714338841316202
};
\addplot [very thick, color2]
table {%
0 0.889808917719106
0.1 0.884152771485075
0.15 0.885038393073814
0.2 0.883365336885207
0.25 0.888204855935001
0.3 0.887783304717646
0.35 0.88450964415707
0.4 0.869616640357278
0.45 0.885851301224168
0.5 0.888155830493167
0.55 0.887385809839144
0.6 0.886501104208304
0.65 0.888719026701249
0.7 0.88744155932349
0.75 0.886928966106214
0.8 0.889723331607986
0.85 0.888327488047893
0.9 0.889338091868981
0.95 0.888147215967596
1 0.891788212151924
1.2 0.892298334114907
1.4 0.876786178430736
1.6 0.881670114703315
1.8 0.879870083083938
2 0.886394470469215
2.2 0.882721828365032
2.4 0.885548994647849
2.6 0.877130546678725
2.8 0.886355376721568
3 0.883958584263239
3.2 0.88448771322927
3.4 0.880653502320358
3.6 0.8797441223312
3.8 0.878059690695617
4 0.872566015037349
4.2 0.871987353966362
4.4 0.87093406204138
4.6 0.862251063599375
4.8 0.863606292988379
5 0.864026752086257
5.2 0.86900431542515
5.4 0.856361609818968
5.6 0.843336793757335
5.8 0.858571337572182
};
\end{axis}

\end{tikzpicture} \quad
    \begin{tikzpicture}

\definecolor{color0}{rgb}{0.12156862745098,0.466666666666667,0.705882352941177}
\definecolor{color1}{rgb}{1,0.498039215686275,0.0549019607843137}
\definecolor{color2}{rgb}{0.172549019607843,0.627450980392157,0.172549019607843}

\begin{axis}[
tick align=outside,
tick pos=left,
title={500 circuits},
x grid style={white!69.0196078431373!black},
xlabel={Link fuzziness [factor]},
xmajorgrids,
xmin=0, xmax=6,
xtick style={color=black},
y grid style={white!69.0196078431373!black},
ylabel={Average fairness index},
ymajorgrids,
ymin=0, ymax=1,
ytick style={color=black}
]
\addplot [very thick, color0]
table {%
0 0.401427849358274
0.2 0.407594590393813
0.6 0.415477691471691
0.8 0.427874964460325
1 0.372991156847337
1.2 0.415631661312989
1.4 0.413722956625339
1.6 0.355701468422039
1.8 0.417287257404126
2 0.360824266971918
2.2 0.376415744456884
2.4 0.399093597900667
2.6 0.364451136345897
2.8 0.337570705060817
3 0.362004498372987
3.2 0.389133104352056
3.4 0.389134115327178
3.6 0.390358821740129
3.8 0.374626803885227
4 0.389245526754649
4.2 0.364823566106279
4.4 0.394413151312864
4.6 0.385465825355419
4.8 0.411230477598467
5 0.382926551696808
5.2 0.397845859959923
5.4 0.394527283270623
5.6 0.392624451654936
5.8 0.459323248356592
};
\addplot [very thick, color1]
table {%
0 0.352221361665069
0.2 0.380967668305008
0.6 0.333818099268914
0.8 0.356607280185873
1 0.328205188512186
1.2 0.341285631208057
1.4 0.326491271196965
1.6 0.326520434359689
1.8 0.355908100618104
2 0.374144954092062
2.2 0.375259152130801
2.4 0.383923820241048
2.6 0.377268165194989
2.8 0.400494254676775
3 0.370659046003033
3.2 0.34471053512311
3.4 0.363675186209682
3.6 0.365867183868336
3.8 0.381908802176369
4 0.326311332273672
4.2 0.359751677695726
4.4 0.410641731534355
4.6 0.35735023771741
4.8 0.366865793622057
5 0.346385286006268
5.2 0.390002499156756
5.4 0.386317467731017
5.6 0.366326463234885
5.8 0.363748203820027
};
\addplot [very thick, color2]
table {%
0 0.859080591227445
0.2 0.861221531946489
0.6 0.868843575793966
0.8 0.868404221796631
1 0.869014817597268
1.2 0.869840992758482
1.4 0.870017507001288
1.6 0.870761686424372
1.8 0.872378963614057
2 0.874224124059821
2.2 0.873917318233745
2.4 0.87398223575878
2.6 0.875033566641148
2.8 0.874883012603564
3 0.876009970072391
3.2 0.875794072355084
3.4 0.877700648562953
3.6 0.880218689543263
3.8 0.87819498152723
4 0.879719232842457
4.2 0.880035463511116
4.4 0.881144094467828
4.6 0.881503718768193
4.8 0.880364622680062
5 0.88035764606896
5.2 0.878552894683505
5.4 0.876432385010141
5.6 0.870307418395356
5.8 0.860269196461685
};
\end{axis}

\end{tikzpicture}
    \caption{Impact of link latency deviation from the system model on achieved fairness, in random networks with 50~relays. Fairness is measured in terms of the fairness index $F$.}
    \label{fig:eval-fuzziness-fairness}
}
\end{figure}
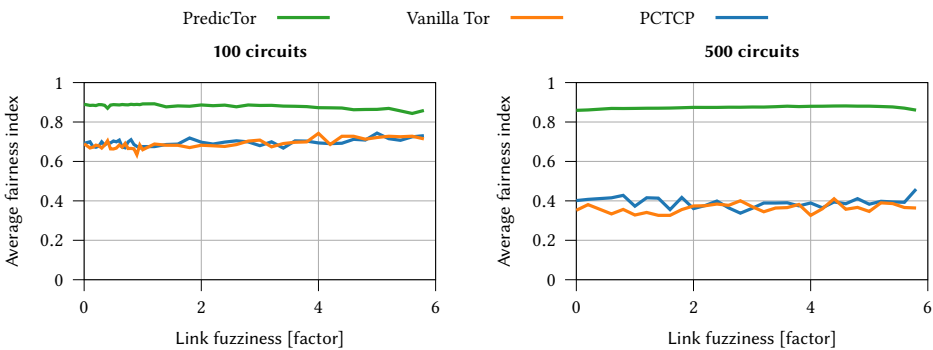

We first consider the achieved latency, shown in Figure~\ref{fig:eval-fuzziness-latency}.
Recall that in both cases (more and less congestion),
PredicTor achieved much lower latency than vanilla Tor and PCTCP
if the link latencies exactly matched the system model, as presented in Section~\ref{sec:network-complexity}.
Introducing link fuzziness now creates a more differentiated picture.
The first observation that can be made is that, for all of the considered Tor variants,
the overall latency grows with a growing link fuzziness factor.
This is not surprising due to the aforementioned
shift of the latency distribution for large fuzziness values.
A more significant observation, however, is that
PredicTor is affected by an increasing fuzziness much more
than the traditional approaches in vanilla Tor and PCTCP.
Despite the fact that PredicTor still performs better in this regard,
it loses much of its advantage.

This insight is important for evaluating the suitability of approaches
based on model predictive control for congestion handling:
Such approaches, like PredicTor, heavily rely on the assumption that
their internal system model gives a suitable representation of the real system behavior.
What happens in the case of latency is that PredicTor's predictions about
when data will arrive and what size the buffers will have in the future,
become less accurate with growing link fuzziness.
As a consequence, its effectiveness in reducing latency in the network also declines.
To a certain degree, this issue could be tackled by extending the system model,
including a more complex model for latency
as well as an explicit notion of dealing with latency variations in the controller.
However, on a conceptual level,
the issue of a mismatch between the modeled system behavior and the real behavior cannot fully be avoided.
In this regard, traditional approaches may prove more robust against unexpected external influences.

However, when looking at the achieved throughput and fairness,
presented in Figures~\ref{fig:eval-fuzziness-throughput} and~\ref{fig:eval-fuzziness-fairness},
we can see that a deviation from the system model
does not necessarily degrade performance in every regard.
As can be concluded from the plots,
the achieved throughput and fairness remain relatively unaffected even by large fuzziness values,
even less than the traditional approaches.
On the one hand, this may be due to the fact that we simulate
the distribution of feedback trajectories out-of-band, as discussed before.
On the other hand, however, we attribute this to the fact that
PredicTor operates on the notion of \emph{data~rates}
instead of absolute numbers of packets to be transferred.
These rates can be realized by Tor even if data is available later than expected due to higher latency
or if there are temporary peaks of data due to inaccurate latency prediction.
The controller is only called for computing a data transfer schedule in distinct time intervals.
In the meantime, Tor can adhere to the calculated plan
and benefit from the enforced properties such as max-min fairness,
even if the system model was partly inaccurate.
As a result, PredicTor even proves more robust than traditional approaches in these specific regards.
Please also note that, in its current form,
PredicTor makes use of TCP for realizing the underlying data transfer.
While this design choice was primarily made for simplicity reasons,
we now see that it also helps with robustness:
If PredicTor also ran its own transport protocol based on its system model,
the impact of a system model mismatch might have been more severe.
We thus think that it may constitute a promising strategy to combine predictive control approaches
with traditional algorithms in a way similar to what we did in PredicTor.

\subsection{Impact of Traffic Patterns} \label{sec:eval-traffic-patterns}

In order to better understand PredicTor's behavior with regard to its traffic dynamics,
we now focus on how it deals with other \emph{traffic patterns}.
The behavior and effectiveness of every congestion control algorithm
clearly also depends on the kind of traffic it is supposed to handle.

In the previous subsections, we only considered \emph{bulk traffic}.
That is, the data to be transferred denotes an infinite stream of bytes.
The intention thereof was to analyze the steady state behavior,
which provided general insights on PredicTor's mechanics.
At the same time, it enabled us to accurately measure the achieved throughput.
This approach, however, is not sufficient for establishing an understanding
of PredicTor's dynamic behavior.
We therefore consider another scenario with more dynamic traffic,
where data streams come and go.

\begin{figure}
    \pgfplotsset{width=.8\textwidth,height=.4\textwidth}
    \pgfplotsset{
    tick label style={font=\scriptsize\sffamily},
    label style={font=\scriptsize\sffamily},
    title style={font=\scriptsize\bfseries\sffamily},
    legend style={font=\footnotesize},
    group/horizontal sep=0.5cm,
    }
    \begin{tikzpicture}[font=\scriptsize]
    \node (v) {PredicTor (web)};
    \draw[color=color2,very thick] ($(v.east) + (0.15cm,0cm)$) -- +(0.7cm,0cm);
    \node[below=1em of v.west,anchor=west] (v2) {PredicTor (bulk)};
    \draw[color=color2,very thick,dashed] ({$(v.east) + (0.15cm,0cm)$} |- {v2.east}) -- +(0.7cm,0cm);

    \node[right=1.75cm of v] (s) {Vanilla Tor (web)};
    \draw[color=color1,very thick] ($(s.east) + (0.15cm,0cm)$) -- +(0.7cm,0cm);
    \node[below=1em of s.west,anchor=west] (s2) {Vanilla Tor (bulk)};
    \draw[color=color1,very thick,dashed] ({$(s.east) + (0.15cm,0cm)$} |- {s2.east}) -- +(0.7cm,0cm);

    \node[right=1.75cm of s] (n) {PCTCP (web)};
    \draw[color=color0,very thick] ($(n.east) + (0.15cm,0cm)$) -- +(0.7cm,0cm);
    \node[below=1em of n.west,anchor=west] (n2) {PCTCP (bulk)};
    \draw[color=color0,very thick,dashed] ({$(n.east) + (0.15cm,0cm)$} |- {n2.east}) -- +(0.7cm,0cm);
    \end{tikzpicture} \par %

    \includegraphics{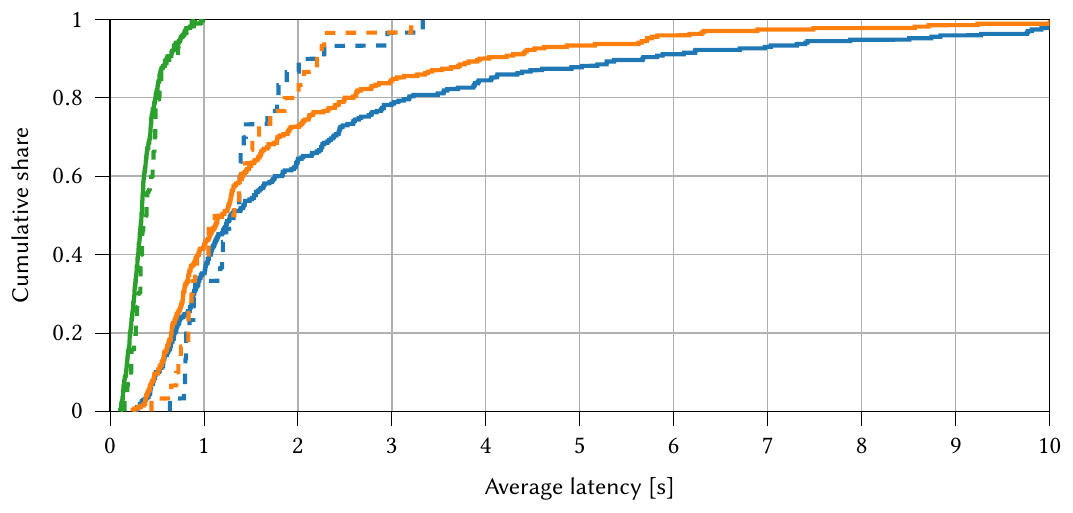}
	\caption{Circuit latencies with mixed bulk/web traffic (50~relays,
    300~circuits, 90\%~web circuits).}
	\label{fig:webratio}
\end{figure}
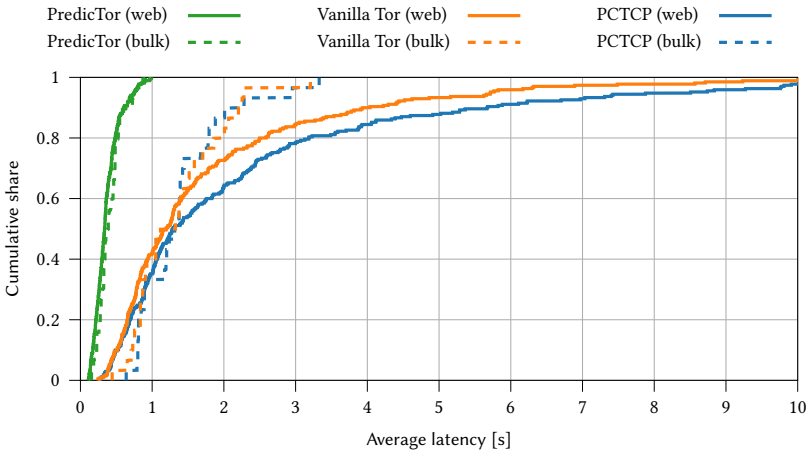

In order to do so, we follow a simple methodology that was put forward in~\cite{jansen2012methodically}
and since then has been applied by a series of publications in this field~%
\cite{DBLP:conf/ndss/WacekTBS13,DBLP:conf/uss/JansenGWSS14,tschorsch16bktap}.
The general approach is to divide the circuits into two groups:
On the one hand, a certain fraction of the circuits carries out bulk data transfers,
similarly to our previous approach.
On the other hand, the other circuits are regarded as \emph{interactive web} circuits.
Their behavior is meant to mimic that of a client interactively browsing the web.
More specifically, such circuits transfer an object 320~KB in size
and wait for a random amount of time (between 1~and 2~seconds) before they repeat.
This way, a certain amount of traffic volatility is created.
Although this model is very simplistic in nature,
we rely on it to establish comparability with previous work in the field.
In our implementation, the bulk circuits start first and the web circuits join later at random times.
We focus on one representative example for explaining various behavioral aspects that can be observed.
In particular, we continue to set the number of relays to~50 and the number of circuits to~300,
which corresponds to the mean value between the previous examples for low and high degrees of congestion.
Considering even more congested networks
would run counter to our goal of simulating interactive circuits
because even smaller per-circuit data rates would make bulk and interactive traffic more similar
due to the increased transfer times.
Moreover, we choose a share of interactive web circuits as~90\%
to approximate the estimation from previous work~\cite{DBLP:conf/pet/McCoyBGKS08,tschorsch16bktap}.
We carry out 25~random repetitions of this simulation scenario
and measure the byte-wise latency of data through the network, as before.

Figure~\ref{fig:webratio} presents a CDF plot of the achieved per-circuit latencies over all runs,
differentiating between bulk and web circuits, for PredicTor as well as vanilla Tor and PCTCP.
The main observation that can be made is that PredicTor achieves low latency for its circuits
even in this scenario with 90\%~web circuits.
In contrast, the majority of circuits handled by the traditional congestion control algorithms
exhibit clearly worse byte-wise latency.
Also, for web circuits, they lead to a long tail of extremely high latency.
This is due to the fact that the web circuits join the network
when it is already overloaded and contains large queues.

One might have expected PredicTor to perform worse,
because the short flows give it less opportunity to apply its predictive behavior.
However, there are two main reasons why this is not the case:
Firstly, even these short flows cover several optimization time steps of PredicTor
so it can in fact apply its predictions to some extent.
Secondly, and even more importantly,
this experiment clearly visualizes PredicTor's second characteristic behavioral trait,
apart from predictiveness---its cautiousness or \emph{pessimistic scheduling}.
While vanilla Tor and PCTCP send as much data into the network as possible,
PredicTor only does so when the circuits are assigned an appropriate data rate by the optimizer.
This, in return, only happens if the network is in fact
able to promptly process and forward the data.
Put differently, we can again see the trade-off made in PredicTor:
It optimizes the latency that payload bytes in the network experience,
at the cost of sacrificing throughput, as we have shown in Section~\ref{sec:network-complexity}.
In the following, we discuss, among others, this relationship more in-depth.

\subsection{Discussion} \label{sec:discussion}

The different steps of our evaluation convey the following overall picture:
PredicTor is highly effective at realizing what is explicitly defined
within its formal optimization objective.
In particular, the improvements it achieves with regard to latency and fairness
compared to vanilla~Tor and PCTCP, are considerable.
While this already becomes apparent by looking at absolute measurement values from selected runs,
the most remarkable difference lies in its asymptomatic behavior:
Unlike vanilla~Tor and PCTCP, both metrics do not degrade with growing levels of congestion,
but stay nearly constant independently of the level of congestion,
due to the created backpressure.
However, we have also seen that PredicTor achieves lower throughput than the traditional approaches.
It therefore makes a clear trade-off that is defined by the optimization goal in the controller.
Not being able to simultaneously optimize throughput and latency
is not a shortcoming specific to PredicTor,
but has long been known as an inherent limitation of congestion control algorithms~%
\cite{DBLP:journals/tcom/Jaffe81a}.
While the traditional approaches we considered (vanilla Tor and PCTCP) optimize for throughput,
PredicTor puts the emphasis on latency instead.
Other strategies would include, \eg, alternating between these goals over time,
as is done by BBR~\cite{DBLP:journals/queue/CardwellCGYJ16}.

We also showed that the latency improvement is not an artifact of the lower throughput,
but an achievement of the controller itself.
The explicit queue constraints~%
\eqref{eq:mpc_full_optim_sc_max}--\eqref{eq:mpc_full_optim_shat_max}
in the optimization problem enforce low backlog
and thus a reduced aggressiveness of the generated traffic.
As a consequence, we consider PredicTor and similar approaches based on distributed MPC
to bear strong potential as the base for novel congestion control mechanisms
that achieve performance values and trade-offs
that are not yet covered by existing traditional approaches.

Apart from this trade-off, we have seen two major disadvantages:
Firstly, like all MPC-based approaches,
PredicTor is dependent on the mathematical system model
it utilizes for making predictions of the network state.
Our evaluation reveals that deviations from this model can severely degrade the performance.
Moreover, the consequences of such model mismatches are not necessarily easy to foresee in advance.
As an example, we saw that throughput and fairness remained relatively unaffected
by growing link latency model errors.
We have also seen that the combination of MPC-based approaches like PredicTor
with traditional underlay transport protocols like TCP can be beneficial with regard to robustness.
The second factor that may potentially prove disadvantageous concerns scalability.
In our naive implementation, computational effort and communication overhead
grow linearly with the number of participating relays.
A multitude of improvement steps could be considered to alleviate or overcome this issue:
For instance, we imagine solving only individually reduced optimization problems at each relay
instead of the complete optimization problem we presented here.
Also, parameters like data resolution, horizon length, and data representation%
---including data compression---should be taken into account.
Continuing to research such refinements can make MPC-based congestion control schemes
very interesting alternatives to existing algorithms.

\section{Related Work} \label{sec:related-work}

Efficiently transferring the circuits' data through the Tor network is far from trivial.
There is a multitude of factors that is known for contributing to
performance issues in Tor~\cite{DBLP:journals/csur/AlSabahG16}.
These include the circuit selection~\cite{congestion-tor12},
local handling of connections at each relay~\cite{DBLP:conf/uss/JansenGWSS14}
and the transport protocol~\cite{udp-or}.
Congestion control touches each of these fields.
Research has shown that insufficient congestion control is a major factor
for Tor's performance problems~\cite{DBLP:journals/csur/AlSabahG16}.
Despite years of research on the topic,
many of these problems remain as of today,
also due to the challenging deployment process
for fundamental changes to the Tor network
\cite{DBLP:journals/corr/abs-1709-01044,doepmann18deployingonions}.

Since Tor is an overlay network,
there generally are multiple conceivable approaches towards congestion control.
In particular, congestion control could either be carried out
\emph{end-to-end} or \emph{hop-by-hop}.
Operating end-to-end matches more closely the classical notion of congestion control
as it is commonly understood for underlay networks.
In Tor, this would mean that only the endpoints (client and exit) are involved.
In fact, this is how vanilla Tor currently operates.
Contrary to many other IP~networks, though,
reliability is currently implemented in a hop-by-hop manner between Tor relays.
Several previous proposals have decided to stick to the paradigm of end-to-end congestion control.
For example, UDP-OR~\cite{udp-or} tunnels a single TCP connection through Tor
using UDP as an underlay.
IPPriv~\cite{DBLP:conf/icc/KiralyC09} follows a similar approach using IPsec.
Taking the idea of \enquote{stateless} intermediate relays one step further,
one could even consider applying active queue management techniques like CoDel~\cite{codel}
or quality~of~service approaches (\eg DiffServ~\cite{rfc2474} or DPS~\cite{stoica1999dps}).
The anonymization functionality could even be moved completely to the network stack,
as LAP~\cite{hsiao2012lap}, HORNET~\cite{chen2015hornet}, and TARANET~\cite{chen2018taranet} demonstrate.
However, there is a common drawback for all of these approaches:
Since the relays within a circuit may be located all over the world,
the resulting round-trip times between the endpoints become very large.
As a consequence, the increased feedback loop
results in degraded performance~\cite{DBLP:conf/dsn/AmirD03,tschorsch12transport}.
PredicTor therefore takes advantage of the fact
that the intermediate relays operate on the application layer anyways,
so they can be taken into account for congestion control,
keeping the feedback loop small and potentially enabling better performance.
This strategy has been followed before, \eg, by replacing Tor's very coarse-grained
end-to-end congestion window with a more flexible scheme~\cite{DBLP:conf/pet/AlSabahBGGMSV11},
or even integrating multi-hop congestion control
into a tailored transport protocol~\cite{tschorsch16bktap}.

On the other hand, PCTCP~\cite{alsabah2013pctcp},
which uses a dedicated TCP connection between each relay for every circuit,
has the potential to be actually deployed in Tor.
While PCTCP provides some improvements, \eg, in fairness,
it still does not provide sufficient congestion control.
Other approaches often require changes to the network infrastructure
and are therefore not directly applicable.
These approaches have focused primarily on re-using existing approaches
from the networking research for this specific use case.
In contrast, our work facilitates recent advances in the field of control technology.
Thereby, we open a new perspective for advancing congestion control using interdisciplinary research.

Multi-hop congestion control has been an active research topic in other fields as well.
For example, numerous scientific contributions apply suitable schemes
in the context of (wireless) mesh networks~%
\cite{hop-by-hop-wireless,DBLP:journals/adhoc/ScheuermannLM08,DBLP:conf/mobicom/JiangZW09}.
These, however, have slightly different use cases and premises.
For example, in Tor, it would not be acceptable to route around
congested areas of the network as this would put the user's anonymity to risk.
However, in special situations,
our approach may also be applicable to these scenarios in some similar form.

The problem of congestion in networks has also been studied extensively
from a control theoretical perspective in the past.
Previous works include classic linear control~\cite{Mascolo1999}
including PID~\cite{Yanfie2003} and state-feedback LQR control~\cite{Azuma2006}.
It is well understood
that delay is among the main challenges of controlling the network.
More recently, especially optimization-based methods have been applied
to the problem with promising results~\cite{He2007,Mota2012}.
Model predictive control (MPC), as applied in~\cite{Mota2012},
is an advanced control technique that can deal with non-linear systems and
explicitly take constraints into consideration.
Its predictive control action is particularly suited for systems with significant delay.
Furthermore, MPC has received significant attention
as a method for distributed control~\cite{Negenborn2014, CHRISTOFIDES2013},
where local controllers interact to jointly control an interconnected system.
Distributed MPC is often applied to systems with a complex network character,
such as transportation systems~\cite{Dunbar2012},
energy management~\cite{Patel2016}
or process industry applications~\cite{CHRISTOFIDES2013},
where a centralized solution is prohibitive due to the size of the system
or privacy concerns.
In order to obtain global properties, the local action is often coordinated by exchanging information about predicted future behavior~\cite{Negenborn2014}.

\section{Conclusion and Future Work} \label{sec:conclusion}
In this work, we proposed a refined version of PredicTor,
a novel model predictive control formulation to tackle the challenge of
congestion control in multi-hop overlay networks like Tor.
PredicTor is a distributed approach that relies on exchanging information
about the predicted network state between adjacent nodes.

We presented a thorough evaluation of PredicTor's performance in complex networks.
Our results indicate that approaches like PredicTor
that build upon distributed model predictive control for congestion control
can, in fact, achieve clear improvements in various regards.
The flexibility of tailoring the optimization goal to the exact requirements of a use case
makes it possible to realize flexible trade-offs.
In PredicTor, we chose to prioritize a strict notion of max-min fairness
as well as low latency in the network.
As our evaluation shows, PredicTor is able to clearly improve on the status quo in these regards.
On the other hand, these benefits are traded for lower throughput.

Our work on using model predictive control for congestion handling
shows the potential of bringing together traditional networking research
and modern control theoretical approaches.
However, our work only denotes a starting point for this research direction.
Several open issues remain:
One important current drawback is the dependency on an accurate system model.
If this underlying description of the system does, \eg, not capture an unexpected external influence,
the advantage of model predictive control can easily be lost.
Another research question that remains unanswered for now is scalability:
Our implementation of PredicTor was not optimized in this regard
and we did not take communication overhead into account
because the resulting goodput ratio is too dependent on the real network conditions
(\ie the absolute data rate).
However, the computational effort and communication
necessary for exchanging the feedback information
currently grows linearly in the number of relays for each node.
This puts a natural upper limit on the size of networks that can efficiently be handled.
However, for overlay networks that are limited in size, it may prove viable.
Future research will have to show
whether these disadvantages can be alleviated to enable real-world application.

\begin{acks}
This work has been partially funded by the
Deutsche Forschungsgemeinschaft
(DFG, German Research Foundation, TS 477/1-1).
\end{acks}
\printbibliography

@Article{Azuma2006,
  author   = {Azuma, Takehito and Fujita, Tsunetoshi and Fujita, Masayuki},
  title    = {Congestion control for {TCP/AQM} networks using state predictive control},
  journal  = {Electrical Engineering in Japan},
  year     = {2006},
  volume   = {156},
  number   = {3},
  pages    = {41--47},
  doi      = {10.1002/eej.20405},
}

@InProceedings{Mota2012,
  author   = {Mota, Joao F.C. and Xavier, Joao M.F. and Aguiar, Pedro M.Q. and Puschel, Markus},
  title    = {Distributed {ADMM} for model predictive control and congestion control},
  booktitle  = {Proceedings of the 51st IEEE Conference on Decision and Control},
  year     = {2012},
  pages    = {5110--5115},
  doi      = {10.1109/CDC.2012.6426141},
  isbn     = {978-1-4673-2066-5},
}

@inproceedings{alsabah2013pctcp,
  title={{PCTCP}: Per-circuit {TCP-over-IPsec} transport for anonymous communication overlay networks},
  author={AlSabah, Mashael and Goldberg, Ian},
  crossref = {ccs13},
}

@inproceedings{tor,
  author    = {Roger Dingledine and
               Nick Mathewson and
               Paul F. Syverson},
  title     = {Tor: The Second-Generation Onion Router},
  booktitle = {Proceedings of the 13th {USENIX} Security Symposium},
  pages     = {303--320},
  year      = {2004},
}

@misc{nstor,
  author={Florian Tschorsch},
  title={nstor},
  year = {2017},
  url="https://github.com/tschorsch/nstor"
}

@book{bertsekas1992data,
  title={Data networks},
  author={Bertsekas, Dimitri P and Gallager, Robert G and Humblet, Pierre},
  edition={2nd edition},
  year={1992},
  publisher={Prentice-Hall International New Jersey}
}

@inproceedings{DBLP:conf/p2p/DhungelSRHR10,
  author    = {Prithula Dhungel and
               Moritz Steiner and
               Ivinko Rimac and
               Volker Hilt and
               Keith W. Ross},
  title     = {Waiting for Anonymity: Understanding Delays in the {Tor} Overlay},
  booktitle = {Proceedings of the 10th {IEEE} Conference on Peer-to-Peer Computing},
  pages     = {1--4},
  year      = {2010},
}

@Article{Andersson2018,
  author   = {Andersson, Joel A. E. and Gillis, Joris and Horn, Greg and Rawlings, James B. and Diehl, Moritz},
  journal  = {Mathematical Programming Computation},
  title    = {{CasADi: A software framework for nonlinear optimization and optimal control}},
  volume = {11},
  year     = {2018},
  issn     = {1867-2949},
  doi      = {10.1007/s12532-018-0139-4},
  isbn     = {1886529000},
  pmid     = {2091171},
}

@Article{Andreas2006,
  author   = {W{\"a}chter, Andreas and Biegler, Lorenz T.},
  journal  = {Mathematical Programming},
  title    = {On the implementation of an interior-point filter line-search algorithm for large-scale nonlinear programming},
  year     = {2006},
  issn     = {1436-4646},
  number   = {1},
  pages    = {25--57},
  volume   = {106},
  day      = {01},
  doi      = {10.1007/s10107-004-0559-y},
}

@InProceedings{Mascolo1999,
  author       = {Mascolo, Saverio},
  booktitle    = {Proceedings of the 38th IEEE Conference on Decision and Control},
  title        = {Classical control theory for congestion avoidance in high-speed internet},
  date         = {1999-12},
  pages        = {2709--2714},
}

@InProceedings{Yanfie2003,
  author       = {Yanfie, Fan and Fengyuan, Ren and Chuang, Lin},
  booktitle    = {Proceedings of the 8th IEEE Symposium on Computers and Communications.},
  title        = {Design a {PID} controller for active queue management},
  year         = {2003},
  pages        = {985--990},
}

@Article{He2007,
  author    = {He, Jiayue and Bresler, Ma'ayan and Chiang, Mung and Rexford, Jennifer},
  journal   = {Journal on Selected Areas in Communications},
  title     = {Towards robust multi-layer traffic engineering: Optimization of congestion control and routing},
  year      = {2007},
  number    = {5},
  pages     = {868--880},
  volume    = {25},
}

@Article{Negenborn2014,
  author    = {R. R. {Negenborn} and J. M. {Maestre}},
  journal   = {IEEE Control Systems Magazine},
  title     = {Distributed Model Predictive Control: An Overview and Roadmap of Future Research Opportunities},
  year      = {2014},
  issn      = {1941-000X},
  number    = {4},
  pages     = {87-97},
  volume    = {34},
  doi       = {10.1109/MCS.2014.2320397},
}

@ARTICLE{Dunbar2012,
author={W. B. {Dunbar} and D. S. {Caveney}},
journal={IEEE Transactions on Automatic Control},
title={Distributed Receding Horizon Control of Vehicle Platoons: Stability and String Stability},
year={2012},
volume={57},
number={3},
pages={620-633},
doi={10.1109/TAC.2011.2159651},
ISSN={2334-3303},}

@InProceedings{Patel2016,
  author    = {N. R. {Patel} and M. J. {Risbeck} and J. B. {Rawlings} and M. J. {Wenzel} and R. D. {Turney}},
  booktitle = {2016 American Control Conference (ACC)},
  title     = {Distributed economic model predictive control for large-scale building temperature regulation},
  year      = {2016},
  month     = {7},
  pages     = {895-900},
  doi       = {10.1109/ACC.2016.7525028},
  issn      = {2378-5861},
}

@Article{CHRISTOFIDES2013,
title = "Distributed model predictive control: A tutorial review and future research directions",
journal = "Computers \& Chemical Engineering",
volume = "51",
pages = "21 - 41",
year = "2013",
author = "Panagiotis D. Christofides and Riccardo Scattolini and David Mu\~{n}oz de la Pe\~{n}a and Jinfeng Liu",
}

@MISC{tor-metrics,
  AUTHOR = {{The Tor Project}},
  URL = {https://metrics.torproject.org/},
  TITLE = {Tor Metrics},
  year = {2021},
}

@PROCEEDINGS{ccta2020,
  LOCATION = {Montreal, QC, Canada},
  BOOKTITLE = {CCTA~'20: {IEEE} Conference on Control Technology and Applications},
  DATE = {2020-08},
}

@inproceedings{PredicTor,
  author    = {Felix Fiedler and
               Christoph Döpmann and
               Florian Tschorsch and
               Sergio Lucia},
  title     = {PredicTor: Predictive Congestion Control for the Tor Network},
  pages     = {863--870},
  crossref  = {ccta2020},
}

@inproceedings{DBLP:conf/uss/ReardonG09,
  author    = {Joel Reardon and
               Ian Goldberg},
  title     = {Improving {Tor} using a {TCP-over-DTLS} Tunnel},
  crossref  = {usenixsecurity09},
}

@PROCEEDINGS{usenixsecurity21,
  BOOKTITLE = {USENIX Security~'21: Proceedings of the 30th USENIX Security Symposium},
  DATE = {2021-08},
}

@inproceedings{neverenough-sec2021,
  author = {Rob Jansen and Justin Tracey and Ian Goldberg},
  title = {Once is Never Enough: Foundations for Sound Statistical Inference in {Tor} Network Experimentation},
  crossref = {usenixsecurity21},
}

@inproceedings{DBLP:conf/pet/AlSabahBGGMSV11,
  author    = {Mashael AlSabah and
               Kevin S. Bauer and
               Ian Goldberg and
               Dirk Grunwald and
               Damon McCoy and
               Stefan Savage and
               Geoffrey M. Voelker},
  title     = {DefenestraTor: Throwing Out Windows in Tor},
  pages     = {134--154},
  crossref  = {pet11},
}

@article{DBLP:journals/csur/AlSabahG16,
  author    = {Mashael AlSabah and
               Ian Goldberg},
  title     = {Performance and Security Improvements for Tor: {A} Survey},
  journal   = {{ACM} Computing Surveys},
  volume    = {49},
  number    = {2},
  pages     = {32:1--32:36},
  date      = {2016-09},
}

@inproceedings{congestion-tor12,
  title = {{Congestion-aware Path Selection for Tor}},
  author = {Tao Wang and Kevin Bauer and Clara Forero and Ian Goldberg},
  booktitle = {FC'12: Proceedings of Financial Cryptography and Data Security},
  date = {2012-02},
}

@inproceedings{DBLP:conf/uss/JansenGWSS14,
  author    = {Rob Jansen and
               John Geddes and
               Chris Wacek and
               Micah Sherr and
               Paul F. Syverson},
  title     = {Never Been {KIST:} Tor's Congestion Management Blossoms with Kernel-Informed
               Socket Transport},
  pages     = {127--142},
  crossref = {usenixsecurity14},
}

@inproceedings{udp-or,
  author    = {Camilo Viecco},
  title     = {UDP-OR: A Fair Onion Transport Design},
  crossref = {hotpets08},
}

@techreport{jains-fairness-index,
  author = {Rajendra K. Jain and Dah-Ming W. Chiu and William R. Hawe},
  title = {A Quantitative Measure of Fairness and Discrimination for Resource Allocation in Shared Computer Systems},
  type = {DEC Research Report TR-301},
  institution = {Digital Equipment Corporation},
  date = {1984-09},
  pages = {38},
}

@ARTICLE{hop-by-hop-wireless,
  author={Y. {Yi} and S. {Shakkottai}},
  journal={IEEE/ACM Transactions on Networking},
  title={Hop-by-Hop Congestion Control Over a Wireless Multi-Hop Network},
  year={2007},
  volume={15},
  number={1},
  pages={133-144},
}

@article{DBLP:journals/adhoc/ScheuermannLM08,
  author    = {Bj{\"{o}}rn Scheuermann and
               Christian Lochert and
               Martin Mauve},
  title     = {Implicit hop-by-hop congestion control in wireless multihop networks},
  journal   = {Ad Hoc Networks},
  volume    = {6},
  number    = {2},
  pages     = {260--286},
  year      = {2008},
}

@inproceedings{DBLP:conf/mobicom/JiangZW09,
  author    = {Shengming Jiang and
               Qin Zuo and
               Gang Wei},
  title     = {Decoupling congestion control from {TCP} for multi-hop wireless networks:
               semi-TCP},
  crossref = {challenged-4},
}

@PROCEEDINGS{challenged-4,
  LOCATION = {Beijing, China},
  BOOKTITLE = {Proceedings of the {ACM} workshop on Challenged networks},
  year = {2009}
}

@inproceedings{DBLP:conf/ndss/WacekTBS13,
  author    = {Chris Wacek and
               Henry Tan and
               Kevin S. Bauer and
               Micah Sherr},
  CROSSREF = {ndss13},
  title     = {An Empirical Evaluation of Relay Selection in Tor}
}

@inproceedings{DBLP:conf/pet/McCoyBGKS08,
  author    = {Damon McCoy and
               Kevin S. Bauer and
               Dirk Grunwald and
               Tadayoshi Kohno and
               Douglas C. Sicker},
  title     = {Shining Light in Dark Places: Understanding the Tor Network},
  CROSSREF = {pet08},
}

@article{DBLP:journals/tcom/Jaffe81a,
  author    = {Jeffrey M. Jaffe},
  title     = {Flow Control Power is Nondecentralizable},
  journal   = {{IEEE} Trans. Commun.},
  volume    = {29},
  number    = {9},
  pages     = {1301--1306},
  year      = {1981},
}

@article{DBLP:journals/queue/CardwellCGYJ16,
  author    = {Neal Cardwell and
               Yuchung Cheng and
               C. Stephen Gunn and
               Soheil Hassas Yeganeh and
               Van Jacobson},
  title     = {{BBR:} Congestion-Based Congestion Control},
  journal   = {{ACM} Queue},
  volume    = {14},
  number    = {5},
  pages     = {20--53},
  year      = {2016},
}

@inproceedings{DBLP:conf/dsn/AmirD03,
  author    = {Yair Amir and
               Claudiu Danilov},
  title     = {Reliable Communication in Overlay Networks},
  pages     = {511--520},
  crossref  = {dsn03},
}

@inproceedings{DBLP:conf/icc/KiralyC09,
  author    = {Csaba Kir{\'{a}}ly and
               Renato Lo Cigno},
 title     = {{IPsec}-Based Anonymous Networking: {A} Working Implementation},
  pages     = {1--5},
  crossref  = {icc09},
}

@article{codel,
  author    = {Kathleen M. Nichols and
               Van Jacobson},
  title     = {Controlling Queue Delay},
  journal   = {{ACM} Queue},
  volume    = {10},
  number    = {5},
  pages     = {20},
  year      = {2012}
}

@MISC{stoica1999dps,
  AUTHOR = {Ion Stoica and Hui Zhang and Fred Baker and Yoram Bernet},
  URL = {https://www.ietf.org/archive/id/draft-stoica-diffserv-dps-02.txt},
 title = {Per hop behaviors based on dynamic packet states},
 howpublished = {{IETF} Expired Internet Draft},
 year = {2002}}

@MISC{rfc2474,
  AUTHOR = {Nichols, K. and Blake, S. and Baker, F. and Black, D.},
  INSTITUTION = {RFC Editor},
  LOCATION = {Fremont, CA, USA},
  year = {1998},
  DOI = {10.17487/RFC2474},
  HOWPUBLISHED = {IETF RFC 2474 (Proposed Standard)},
  ISSN = {2070-1721},
  NUMBER = {2474},
  TITLE = {{Definition of the Differentiated Services Field (DS Field) in the IPv4 and IPv6 Headers}},
}

@article{DBLP:journals/corr/abs-1709-01044,
  author    = {Rob Jansen and
               Matthew Traudt},
  title     = {Tor's Been {KIST:} {A} Case Study of Transitioning Tor Research to
               Practice},
  journal   = {CoRR},
  volume    = {abs/1709.01044},
  year      = {2017},
  url       = {http://arxiv.org/abs/1709.01044},
  archivePrefix = {arXiv},
  eprint    = {1709.01044}
}

\end{document}